\documentclass[11pt]{article}
\usepackage[utf8]{inputenc}

\title{}
\author{}
\date{}

\usepackage{amsfonts,amsmath,color,amsthm,amssymb, url, enumerate, bbm}
\usepackage[margin=1in]{geometry}
\usepackage[dvipsnames]{xcolor}
\usepackage{caption}
\usepackage{subcaption}
\usepackage{float}
\usepackage{wrapfig}
\usepackage{tikz, etoolbox}
\usepackage{thmtools, thm-restate}
\usepackage[most]{tcolorbox}
\usetikzlibrary{shapes}
\usetikzlibrary{arrows}
\usetikzlibrary{arrows.meta}
\usetikzlibrary{decorations.markings}

\usepackage{algorithm} 
\usepackage[noend]{algpseudocode}
\usepackage{thm-restate}

\definecolor{navy}{HTML}{40739E}
\definecolor{grey}{HTML}{A0A0A0}

\newtcolorbox{mybox}[2][]{colback=navy!5!white,
colframe=navy!75!black,fonttitle=\bfseries,
colbacktitle=navy!85!black,
title=#2,#1}

\tikzstyle{vertex}=[circle, draw, fill, inner sep=0pt, minimum size=5pt]
\newcommand{\vertex}{\node[vertex]}

\tikzstyle{smvertex}=[circle, draw, fill, inner sep=0pt, minimum size=2pt]

\tikzstyle{mvertex}=[circle, draw, fill, inner sep=0pt, minimum size=4pt]

\newcommand{\f}{\frac}

\newcommand{\n}{\f{1}{n}}

\newcommand{\rar}{\rightarrow}

\newcommand{\Z}{\mathbb{Z}}

\newcommand{\N}{\mathbb{N}}
\newcommand{\R}{\mathbb{R}}

\definecolor{purple}{RGB}{128,0,128}

\def\qed{\vbox{\hrule\hbox{\vrule\kern3pt\vbox{\kern6pt}\kern3pt\vrule}\hrule}}

\def\n{\noindent}
\newtheorem{thm}{Theorem}[section]

\newtheorem{cor}[thm]{Corollary}

\newtheorem{prop}[thm]{Proposition}
\newtheorem{defn}[thm]{Definition}
\newtheorem{cm}[thm]{Claim}
\newtheorem{conj}[thm]{Conjecture}

\newtheorem*{thm*}{Theorem}
\newtheorem*{cor*}{Corollary}
\newtheorem*{lm*}{Lemma}
\newtheorem*{cm*}{Claim}
\newtheorem*{prop*}{Proposition}

\theoremstyle{definition}

\newtheoremstyle%
 {Aside}%
 {}{}%
 {\color{purple}\itshape}
 {}%
 {\color{purple}\bfseries}%
 {\color{purple}.}%
 { }{}
\theoremstyle{Aside}

\title{Circulant TSP: Vertices of the Edge-Length Polytope and Superpolynomial Lower Bounds}

\author{Samuel C. Gutekunst}

\begin{document}

\maketitle

\begin{abstract}
    We study the \emph{edge-length} polytope, motivated both by algorithmic research on the \emph{Circulant Traveling Salesman Problem} (Circulant TSP) and number-theoretic research related to the \emph{Buratti-Horak-Rosa conjecture}.  Circulant TSP is a special case of TSP whose overall complexity is a significant still-open question, and where on an input with vertices $\{1, 2, ..., n\}$, the cost of an edge $\{i, j\}$ depends only on its \emph{length} $\min\{|i-j|, n-|i-j|\}$.  The edge-length polytope provides one path to solving circulant TSP instances, and we show that it is intimately connected to the factorization of $n$: the number of vertices scales with $n$ whenever $n$ is prime and with $n^{3/2}$ whenever $n$ is a prime-squared, but there are a superpolynomial  number of vertices whenever $n$ is a power of 2. In contrast, the more-standard Symmetric TSP Polytope has roughly $n!$ vertices.  Hence, for Circulant TSP, a brute-force algorithm checking every vertex is actually efficient in some cases, based on the factorization of $n$.  As an intermediate step, we give superpolynomial lower-bounds on two combinatorial sequences related to the Buratti-Horak-Rosa conjecture, which asks what combinations of edge lengths can comprise a Hamiltonian path. 
\end{abstract}

\section{Introduction}
This main result of this paper is a surprising connection between two strands of research.  First, \emph{circulant TSP} is a special case of the Traveling Salesman Problem (TSP) that has been studied since the 70's, motivated by waste minimization (Garfinkel \cite{Gar77}) and reconfigurable network design (Medova \cite{Med93}).  Circulant TSP has rich connections to number theory: while the overall complexity of circulant TSP has frequently been raised as a significant still-open question   (see, e.g.,  Burkard \cite{Burk97}, Burkard, De\u{\i}neko, Van Dal, Van der Veen, and Woeginger \cite{Burk98}, and Lawler, Lenstra, Rinnooy Kan, and Shmoys  \cite{Law07}),  work in the 70s showed that it could be easily and efficiently solved any time the input number of TSP vertices was a prime; recent work has shown that it is also efficiently solvable anytime the input number of TSP vertices is a prime-squared (Beal, Bouabida, Gutekunst, and Rustad \cite{Bea24}).  

A second strand of research involves number theoretic and combinatorial work related to the  \emph{Buratti-Horak-Rosa conjecture}, which also remains open despite substantial work.    The Buratti-Horak-Rosa conjecture relates to the combinations of edge lengths that can be patched together to form a Hamiltonian path, on a graph where $n$ vertices are arranged around a circle and equally spaced apart.  For instance, consider Figure \ref{fig:BHRconj}, with $n=4$ vertices equally spaced around a circle.  There are only two possible edge-lengths: short edges connecting consecutive vertices, and long edges connecting vertices along a diameter (and because the vertices are equally placed around the circle, e.g., the Euclidean distance between any pair of consecutive vertices is the same).   A Hamiltonian path can either use three short edges, two short edges and one long edge, or one short edge and three long edges.  However, it cannot use three long edges (as such a path, e.g., would never be able to connect the even-indexed vertices to the odd-indexed vertices).

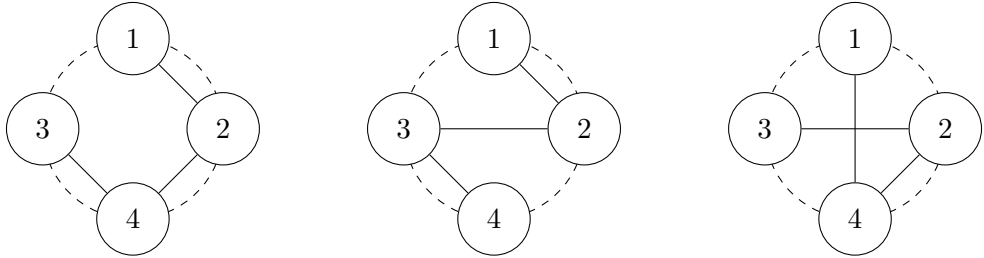
\begin{figure}[h!t]
\centering

\begin{tikzpicture}[scale=0.6]
\tikzset{vertex/.style = {shape=circle,draw,minimum size=2.5em}}
\tikzset{edge/.style = {->,> = latex'}}
\tikzstyle{decision} = [diamond, draw, text badly centered, inner sep=3pt]
\tikzstyle{sq} = [regular polygon,regular polygon sides=4, draw, text badly centered, inner sep=3pt]

\draw[dashed]   (-8, 0) circle[radius=2] node [yshift=0] {};
\node[vertex, fill=white] (a) at  (-10, 0) {$3$};
\node[vertex, fill=white] (b) at  (-6, 0) {$2$};
\node[vertex, fill=white] (c) at  (-8, 2) {$1$};
\node[vertex, fill=white] (d) at  (-8, -2) {$4$};
\draw (c) -- (b);
\draw (b) -- (d);
\draw (d) -- (a);

\draw[dashed]   (0, 0) circle[radius=2] node [yshift=0] {};
\node[vertex, fill=white] (a1) at  (-2, 0) {$3$};
\node[vertex, fill=white] (b1) at  (2, 0) {$2$};
\node[vertex, fill=white] (c1) at  (0, 2) {$1$};
\node[vertex, fill=white] (d1) at  (0, -2) {$4$};
\draw (c1) -- (b1);
\draw (b1) -- (a1);
\draw (a1) -- (d1);

\draw[dashed]   (8, 0) circle[radius=2] node [yshift=0] {};
\node[vertex, fill=white] (a2) at  (6, 0) {$3$};
\node[vertex, fill=white] (b2) at  (10, 0) {$2$};
\node[vertex, fill=white] (c2) at  (8, 2) {$1$};
\node[vertex, fill=white] (d2) at  (8, -2) {$4$};
\draw (c2) -- (d2);
\draw (d2) -- (b2);
\draw (b2) -- (a2);

\end{tikzpicture}

\caption{With four vertices spaced equally around a circle, there are three combinations of edge lengths that can comprise a Hamiltonian path.} \label{fig:BHRconj}
\end{figure} 

Sloane’s On-Line Encyclopedia of Integer Sequences (OEIS) \cite{Slo} sequence A352568 counts the number of combinations of edge lengths that can comprise a Hamiltonian path in this setting; following Figure \ref{fig:BHRconj}, there are 3 combinations when $n=4$. Recent work by McKay and Peters \cite{Mc22} determined the values of A352568 through $n=37$, and showed that the counts are consistent with the Buratti-Horak-Rosa conjecture (also noting that these computations took ``approximately four years of cpu time''!).  This sequence is closely related to  sequence A030077, which only accounts for the total length of the path (using the Euclidean distance).  A030077 lower bounds A352568; McKay and Peters \cite{Mc22} prove that these two sequences match whenever $n$ is prime, twice a prime, or a power of 2.

In this paper, we study a connection between both strands of research.  Our work is motivated by the \emph{edge-length polytope} $EL(n)$ and a related combinatorial sequence, the number of vertices of $EL(n)$.  The edge-length polytope was introduced in  Gutekunst and Williamson \cite{Gut20} as a tool for studying circulant TSP and the TSP more broadly. Our new sequence (the number of vertices of the edge-length polytope) is closely related to A352568.  But while A352568 is monotone-increasing for all known values, the number of vertices of $EL(n)$ is decidedly not.  

The main goal of this paper is to understand the vertices of $EL(n)$: understanding these vertices provides intuition for why current results for the complexity of circulant TSP (so far) seem to rely heavily on the factorization of $n$.  We characterize the edge-length polytope for small values of $n$, in the spirit of early work characterizing the (general) symmetric traveling salesman polytope for small $n$ (see, e.g., Boyd and Cunningham \cite{Boyd91}, Christof,  J{\"u}nger, and Reinelt \cite{ch91}, and Christof and Reinelt \cite{Chr96}). 
 Then we use previous work on circulant TSP to count the number of vertices of $EL(n)$ whenever $n$ is a prime or prime-squared.  We show that the number of vertices of $EL(n)$ scales linearly with $n$ whenever $n$ is prime,  scales with $n^{3/2}$ whenever $n$ is a prime-squared, but  is superpolynomial when $n$ is a power of 2 (specifically, that there are $\Omega\left(n^\frac{\log_2(n)-3}{2}\right)$ vertices when $n$ is a power of 2).  Thus it is possible to define polynomial-time brute-force algorithms for circulant TSP whenever $n$ is a prime or prime-squared, but not when $n$ is a power of 2.  

To get our superpolynomial result about about the vertices of $EL(n)$, we consider Hamiltonian paths as an intermediate step.  We introduce a class of Hamiltonian paths known as \emph{BLG Paths}, inspired by a result of  Bach, Luby, and Goldwasser (cited in Gilmore, Lawler, and Shmoys \cite{Gil85}).  BLG paths arise as the optimal solution to the problem of finding minimum-cost Hamiltonian paths on circulant instances (without a fixed start and end vertex).  We introduce \emph{encoding sequences} as a tool to count BLG paths.  We then show that there are super-polynomially many BLG paths whenever $n$ is a power of 2, each of which has a distinct edge-length vector.  This immediately implies that A352568 grows super-polynomially when $n$ is a power of $2$. Because McKay and Peters \cite{Mc22} showed that A352568 match A030077 when $n$ is a power of 2, our result also immediately implies that A030077  grows super-polynomially.  While our lower-bound is by no means tight, to the best of our knowledge it is the first superpolynomial lower bound on both sequences.

Finally, we note that the edge-length polytope has already shown itself to be a useful tool for polyhedral TSP research.  Specifically,  Gutekunst and Williamson \cite{Gut20} showed that this polytope gave rise to a new facet-defining inequality for the general symmetric Traveling Salesman Problem, and that the classic crown inequalities of Naddef and Rinaldi \cite{Nad92} were fundamentally about the edge-length polytope. 
Substantial TSP research has studied  facet-defining inequalities for the TSP (e.g. the clique-tree inequalities  \cite{Gro86b}, the comb inequalities \cite{Chv73, Gro79}, the crown inequalities \cite{Nad92}, the path inequalities \cite{Cor85}, the path-tree inequalities  \cite{Nad91} the binested inequalities  \cite{Nad92b}, and the rank inequalities \cite{Gro80}). In addition to giving insight into circulant TSP specifically, understanding the edge-length polytope is thus also be of interest for polyhedral approaches to the TSP. 

\subsection{Outline}

In Section \ref{sec:bg}, we begin with background on both the edge-length polytope and the  Buratti-Horak-Rosa conjecture.  Section \ref{sec:el} presents a few results on the edge-length polytope: we describe the edge-length polytope for small $n$ and dervie formulas for the number of vertices when $n$ is prime or prime-squared.  In Section \ref{sec:combo}, we introduce BLG paths.  Counting these paths gives a lower bound on A352568 and A030077, and we show that this lower bound is superpolynmial.  Finally, in Section \ref{sec:lb}, we show how these Hamiltonian paths lower bound the number of vertices of the edge-length polytope.

\section{Background and Notation}\label{sec:bg}

\subsection{Circulant TSP and the Buratti-Horak-Rosa Conjecture}

The Symmetric Traveling Salesman Problem (TSP) is one of the most famous problems in discrete mathematics and theoretical computer science.  An input consists of $n$ vertices $[n]:=\{1, 2, ..., n\}$ and edge costs $c_{ij}=c_{ji}$ (for $1\leq i, j\leq n$), indicating the costs of travelling between vertices $i$ and $j$.  The goal is to find a minimum-cost Hamiltonian cycle, a cycle which visits each vertex exactly once.   With just this set-up, the TSP is well known to be NP-hard  (see, e.g., Theorem 2.9 in Williamson and Shmoys \cite{DDBook}).   Much TSP research has thus focused on special cases, where the edge costs are more structured.  These include the \emph{metric} TSP (where edge costs obey the triangle inequality $c_{ij}+c_{jk} \geq c_{ik}$ for all $i, j, k\in [n]$), the \emph{graph} TSP (where edge costs correspond to distances between vertices in an underlying graph),  the \emph{Euclidean} TSP (where vertices have fixed locations in $\R^{n}$ and edge costs correspond to Euclidean distances between them, and the \emph{(1, 2)-TSP} (where  $c_{ij}\in\{1, 2\}$ for all $i, j$).  See, e.g., \cite{ Aro96, Ber08, Chr76,      kar21, Karp12, Mitch99, Mom16, Muc14, Gha11,  Pap93, Seb14,   Serd78} amoung many others.

In (symmetric) circulant TSP, the edge costs are not necessarily metric, but instead respect circulant symmetry.  Specifically,  define the \emph{length} of an edge $\{i, j\}$ as  $$\ell_{i, j}=\min\{|i-j|, n-|i-j|\}.$$   For instance, edges $\{1, 2\}$, $\{3, 2\}$, and $\{n, 1\}$ all have the same length 1; our short edges in Figure \ref{fig:BHRconj} all have length 1, and the long edges have length 2\footnote{As a technical point, in this paper we use the definition of edge length as $\ell_{i, j}=\min\{|i-j|, n-|i-j|\}.$  The sequence A030077 uses a Euclidean notion of length, i.e. the actual total Euclidean lengths of the paths in Figure \ref{fig:BHRconj} obtained by adding up the lengths of each segment.  Throughout this paper, we only use the former notion of length.}.  Figure \ref{fig:sym} illustrates circulant symmetry more generally.  In circulant TSP, all edges of length $i$ are required to have the same cost $c_i$, so $c_{ij}=c_{\min\{(i-j)\bmod n, (j-i)\bmod n\}}$.  (Here and throughout, when we work mod $n$, we mean with residues in $\{1, ..., n\}$.) Note that this is equivalent to requiring that the matrix of edge costs $C:=(c_{ij})_{i, j=1}^n$ be a symmetric, circulant matrix with $d=\lfloor n/2 \rfloor$ parameters $c_1, c_2, ..., c_d$.  We say that such an input is a \emph{circulant input}, or equivalently, has \emph{circulant costs}.

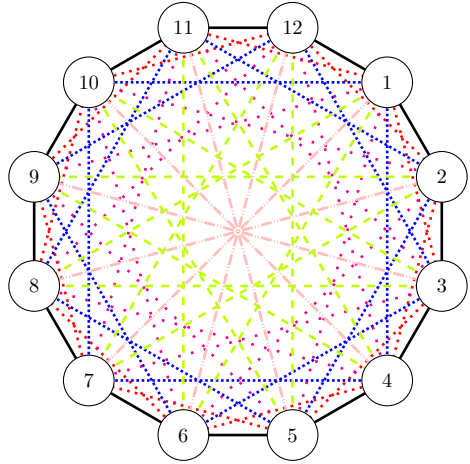
\begin{figure}[t!]
	\begin{center}
\begin{tikzpicture}[scale=0.7, transform shape]
\tikzset{vertex/.style = {shape=circle,draw,minimum size=2.5em}}
\tikzset{edge/.style = {->,> = latex'}}
\node[draw=none,minimum size=8cm,regular polygon,regular polygon sides=12] (a) {};

\foreach [evaluate ={\j=int(mod(\i, 12)+1)}] \i in {1,2, 3, ..., 12}
\draw[line width=1pt] (a.corner \i) -- (a.corner \j);

\foreach [evaluate ={\j=int(mod(\i+1, 12)+1)}] \i in {1,2, 3, ..., 12}
\draw[dotted, line width=1pt, red] (a.corner \i) -- (a.corner \j);

\foreach [evaluate ={\j=int(mod(\i+2, 12)+1)}] \i in {1,2, 3, ..., 12}
\draw[densely dotted, line width=1pt, blue] (a.corner \i) -- (a.corner \j);

\foreach [evaluate ={\j=int(mod(\i+3, 12)+1)}] \i in {1,2, 3, ..., 12}
\draw[loosely dotted, line width=1pt, magenta] (a.corner \i) -- (a.corner \j);

\foreach [evaluate ={\j=int(mod(\i+4, 12)+1)}] \i in {1,2, 3, ..., 12}
\draw[dashed, line width=1pt, lime] (a.corner \i) -- (a.corner \j);

\foreach [evaluate ={\j=int(mod(\i+5, 12)+1)}] \i in {1,2, 3, ..., 12}
\draw[dashdotted, line width=1pt, pink] (a.corner \i) -- (a.corner \j);

\foreach \n [count=\nu from 1, remember=\n as \lastn, evaluate={\nu+\lastn}] in {12, 11, ..., 1} 
\node[vertex, fill=white]at(a.corner \n){$\nu$};

\end{tikzpicture}

	\end{center}
	\caption{Circulant symmetry.  Edges of a fixed length are indistinguishable and have the same cost.  E.g. all edges of the form $\{v, v+1\}$ (where $v+1$ is taken mod $n$) have length 1, and thus the same cost $c_1$. }\label{fig:sym}\end{figure}

One of the reasons for the interest in circulant TSP is that circulant symmetry sometimes, but not always, provides enough structure to make a formally-hard problem easy.  For instance, while the complexity of circulant TSP is open, a classic result of Bach, Luby, and Goldwasser (cited in Gilmore, Lawler, and Shmoys \cite{Gil85}) shows that it is easy to find minimum-cost Hamiltonian paths (without a fixed starting and ending vertex) on circulant instances, and indeed that the greedy algorithm finds such paths.  We formally state this result in Section \ref{sec:HPaths}, but the main idea is as follows: Let $q$ divide $n$, and let $\equiv_n$ denote equivalence mod $n$.    Then the vertices $[n]$ can be partitioned into the sets $$S_1, S_2, ..., S_q=\{1 \leq i \leq n: i \equiv_q 1\}, \{1 \leq i \leq n: i \equiv_q 2\}, ..., \{1 \leq i \leq n: i \equiv_q q\}.$$ Any edge whose length is a multiple of $q$ will have both endpoints within the same $S_i$, so a Hamiltonian path connecting all vertices must use at least $q-1$ edges to connect the sets $S_1, ..., S_q$.  Thus, a Hamiltonian path can use at most $n-q$ edges whose length is a multiple of $q.$  Thus, if a Hamiltonian path uses $t_i$ edges of length $i$ for $i=1,...,d$, it must satisfy $$\sum_{1\leq i\leq d: q|i} t_i \leq n-q, \text{ for each } q \text{ dividing } n.$$  For instance, if $n$ is even, the inequality for $q=2$ says that any Hamiltonian path must satisfy $t_2+t_4+....+t_{n/2} \leq n-2$, i.e. that there is at least 1 edge whose length is odd to connect the sets $S_1=\{1, 3, ..., n-1\}$ and $S_2=\{2, 4, 6, ..., n\}$.

As such, the greediest strategy to find a Hamiltonian path on a circulant instance is to pick the cheapest edge-length (say, $j$), and use as many edges of that length as possible (appealing to the above logic, this turns out to be $n-\gcd(n, j)$ to ensure that the constraint is not violated when $q=\gcd(n, j)$).  Then pick the second-cheapest edge length and use as many of that as possible, and repeat until you have selected $n-1$ edges.  

Finding Hamiltonian paths on circulant instances is an extremal version of the  Buratti-Horak-Rosa problem of characterizing all combinations of edge lengths that can be patched together to form a Hamiltonian path: Buratti-Horak-Rosa ask for \emph{all} combinations, while circulant TSP asks only about combinations that arise when finding minimum-cost Hamiltonian paths with circulant edge costs.  The Buratti-Horak-Rosa conjecture is a striking statement that, in essence, conjectures that the extremal condition for finding minimum-cost Hamiltonian paths is all that matters: 

\begin{conj}[Buratti-Horak-Rosa]\label{conj}
	Let $t_1, t_2, ..., t_d$ be nonnegative integers such that $\sum_{i=1}^d t_i = n-1$ and, for each $q$ that divides $n$, $$\sum_{1\leq i\leq d: q|i} t_i \leq n-q.$$  Then there exists a Hamiltonian path on $[n]$ using $t_i$ edges of length $i$ (for $i=1, ..., d$).
\end{conj}

Both circulant TSP (finding minimum-cost Hamiltonian cycles with circulant edge costs) and the Buratti-Horak-Rosa conjecture (understanding what sets of edge lengths can comprise a Hamiltonian path) have received substantial attention in their own right.  Much of the work on circulant TSP has tackled its complexity in simplified settings:  it took about 15 years to resolve the complexity of even the simplest non-trivial case of circulant TSP (the \emph{two-stripe symmetric circulant TSP}), with work beginning in 2007 with  Greco and Gerace \cite{Grec07} and Gerace and Greco \cite{Ger08b} and with  Gutekunst, Jin, and Williamson \cite{gut22} finally showing that the two-stripe symmetric circulant TSP problem was solvable in polynomial time in 2022.  Another strand of research has instead tried to tackle circulant TSP through the factorization of $n$; in the 70's, Garfinkel \cite{Gar77} showed that circulant TSP could be efficiently solved whenever $n$ was prime, with Beal, Bouabida, Gutekunst, and Rustad \cite{Bea24} recently showing that the same held true whenever $n$ is a prime-squared.  In another direction, de Klerk and Dobre \cite{Klerk11} empirically and theoretically compared lower bounds for circulant TSP. 
See also  \cite{Burk97, Burk98, Gut19b, Gut20, Law07}, amoung others.

Work on the Buratti-Horak-Rosa conjecture dates back to at least the 2000s, when Buratti conjectured a version of Conjecture \ref{conj} whenever $n$ was prime: in that case there are no non-trivial divisors of $n$, and Buratti conjectured that there was a Hamiltonian path with any combination of $n-1$ edges of any length (see Horak and Rosa \cite{Hor09} for a discussion of the origin of the Buratti-Horak-Rosa conjecture).  Since then, substantial number theoretic work has gone into understanding what collections of edge lengths can constitute a Hamiltonian cycle and/or path.  Ollis, Pasotti,  Pellegrini, and Schmitt \cite{Oll22} summarize many of the recent results 
(see also \cite{Avi23, Bur13, Cap10, Cha22, cost18, Din09, Hor09, Oll21, Oll22, Pas14, Pas14b, Pas19, Vaz22}, among others).

A parallel direction of work related to the  Buratti-Horak-Rosa conjecture has been combinatorial, focusing on the number of combinations of edge lengths that can comprise a Hamiltonian path (OEIS sequence A352568), and providing computational support for the Buratti-Horak-Rosa conjecture.  Mariusz Meszka is reported to have verified it through $n=23$ (see McKay and Peters \cite{Mc22}), and McKay and Peters \cite{Mc22} extended that through $n=37$.  McKay and Peters also give values of A352568 for $n=38$ to 50, assuming that the Buratti-Horak-Rosa conjecture is true.

\subsection{Minimum-Cost Hamiltonian Paths in Circulant Graphs}\label{sec:HPaths}
In this section, we build up to a classic result of Bach, Luby, and Goldwasser, which shows how to find minimum-cost Hamiltonian paths with circulant edge costs.  In Section \ref{sec:combo}, we will count these paths and show that they grow super-polynomially.  Throughout this section and the rest of the paper, we consider a graph with $n$ vertices, circulant edge costs, and that all calculations with vertices are implicitly mod $n$ (e.g. if we are at vertex $v$ and take a length-$\ell$ edge to vertex $v+\ell$, we are using $v+\ell$ as shorthand for $(v+\ell) \mod n$).  We also let $d=\lfloor n/2\rfloor$ denote the maximum length of an edge on a circulant instance with $n$ vertices.

First, we define a circulant graph.  

\begin{defn}
Let $S\subseteq \{1, ..., d\}.$ The {\bf circulant graph $C\langle S\rangle$} is the (simple, undirected, unweighted) graph including exactly the edges whose lengths are in $S$.  I.e., the graph with adjacency matrix $$A=(a_{ij})_{i, j=1}^n, \hspace{5mm} a_{ij} = \begin{cases} 1, & (i-j) \bmod n \in S \text{ or } (j-i) \bmod n \in S \\ 0, & \text{ else.} \end{cases}$$
\end{defn}
\noindent For a set of edge-lengths $S$, the graph $C\langle S\rangle$ includes exactly the edges of those lengths in $S$.  See, for example, Figure \ref{fig:CKNotation} for a circulant graph.

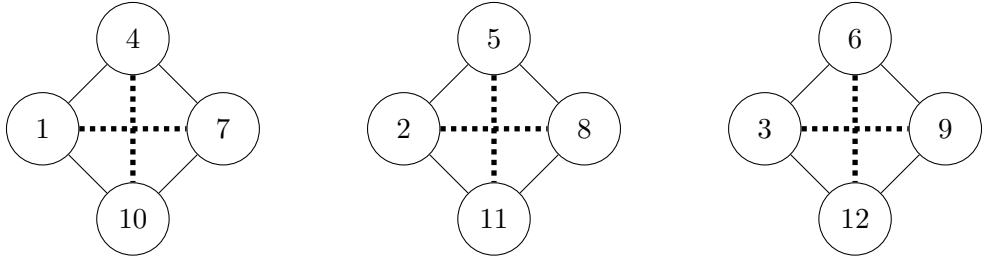
\begin{figure}[t]
\centering

\begin{tikzpicture}[scale=0.6]
\tikzset{vertex/.style = {shape=circle,draw,minimum size=2.5em}}
\tikzset{edge/.style = {->,> = latex'}}
\tikzstyle{decision} = [diamond, draw, text badly centered, inner sep=3pt]
\tikzstyle{sq} = [regular polygon,regular polygon sides=4, draw, text badly centered, inner sep=3pt]
\node[vertex] (a) at  (-10, 0) {$1$};
\node[vertex] (b) at  (-6, 0) {$7$};
\node[vertex] (c) at  (-8, 2) {$4$};
\node[vertex] (d) at  (-8, -2) {$10$};

\node[vertex] (a1) at  (-2, 0) {$2$};
\node[vertex] (b1) at  (2, 0) {$8$};
\node[vertex] (c1) at  (0, 2) {$5$};
\node[vertex] (d1) at  (0, -2) {$11$};

\node[vertex] (a2) at  (6, 0) {$3$};
\node[vertex] (b2) at  (10, 0) {$9$};
\node[vertex] (c2) at  (8, 2) {$6$};
\node[vertex] (d2) at  (8, -2) {$12$};
\draw[dotted,line width=2pt] (a) -- (b);
\draw[dotted,line width=2pt] (c) -- (d);
\draw (c) -- (a);
\draw (d) -- (a);
\draw (d) -- (b);
\draw (c) -- (b);

\draw[dotted,line width=2pt] (a1) -- (b1);
\draw[dotted,line width=2pt] (c1) -- (d1);
\draw (c1) -- (a1);
\draw (d1) -- (a1);
\draw (d1) -- (b1);
\draw (c1) -- (b1);

\draw[dotted,line width=2pt] (a2) -- (b2);
\draw[dotted,line width=2pt] (c2) -- (d2);
\draw (c2) -- (a2);
\draw (d2) -- (a2);
\draw (d2) -- (b2);
\draw (c2) -- (b2);

\end{tikzpicture}

\caption[Circulant graphs and component notation.]{The graph $C\langle\{6, 3\}\rangle$ for $n=12.$  Dashed edges are of length 6.  If $\phi(1)=6$ and $\phi(2)=3$, then $g_0^{\phi}=12, g_1^{\phi}=6$, and $g_2^{\phi}=3.$ } \label{fig:CKNotation}
\end{figure} 

Bach, Luby, and Goldwasser's algorithm can be interpreted as a greedy algorithm that uses as many edges as possible of the cheapest edge-length, and then of the next cheapest, and so on.  Hence, it only depends on the order of the edge-costs $c_1, ...., c_d$, and we define $\phi: [d]\rar [d]$ to be a {\bf cost permutation} that orders the edge costs from least to greatest.  The permutation $\phi$ gives rise to an associated $g$-sequence, which will quantify how greedy we can be.  

\begin{defn}
Let $\phi: [d]\rar [d]$ be a {\bf cost permutation} so that $c_{\phi(1)}\leq c_{\phi(2)} \leq \cdots \leq c_{\phi(d)}.$  The {\bf $g$-sequence associated to $\phi$} is $g^{\phi}=(g_0^{\phi}, g_1^{\phi}, ..., g_d^{\phi}),$ recursively defined by $$g_i^{\phi}=\begin{cases} n, & i=0 \\ \gcd\left(\phi(i), g_{i-1}^{\phi}\right), & \text{ else.} \end{cases}$$
\end{defn}
Equivalently, for $1\leq i\leq d$, $g_i^{\phi}=\gcd(n, \phi(1), \phi(2), ..., \phi(i)).$

The graph $C\langle\{\phi(1), \phi(2), ..., \phi(k)\}\rangle$ (i.e. the circulant graph on vertex set $[n]$ and with exactly the edges of length $\phi(1), \phi(2), ...$ $\phi(k)$) has $g_k^{\phi}$ components with $n/g_k^{\phi}$ vertices each (see, e.g., Burkard and Sandholzer \cite{Burk91}).   Similarly,  let $$\ell :=  \min \{i: 1\leq i\leq d, g_i^{\phi}=1\}.$$  The graph $C\langle\{\phi(1), \phi(2), ..., \phi(k)\}\rangle$ is Hamiltonian if and only if $k\geq \ell$ (again, see Burkard and Sandholzer \cite{Burk91}).

This notation is sufficient to state Bach, Luby, and Goldwasser's result.  The graph $C\langle\{\phi(1)\}\rangle$ has $g_1^{\phi}$ components;  a Hamiltonian path must use at least $g_1^{\phi}-1$ edges that cross these components, and since length-$\phi(1)$ stay within a component, a Hamiltonian path can use at most $(n-1)-(g_1^{\phi}-1)=n-g_1^{\phi}$ of the cheapest length-$\phi(1)$ edges.  Similarly, a Hamiltonian path must can use at most $n-g_2^{\phi}$ total edges edges of length $\phi(1)$ and $\phi(2)$, since at least $g_2^{\phi}-1$ other edges are required to connect the components of  $C\langle\{\phi(1), \phi(2)\}\rangle$.  Thus, in the best case, it would use $n-g_1^{\phi}$ edges of length $\phi(1)$ and $g_1^{\phi}-g_2^{\phi}$ edges of length $\phi(2)$.  Proceeding by a similar argument shows that, in the very best case, a Hamiltonian path would cost $\sum_{i=1}^{\ell} (g_{i-1}^{\phi}-g_i^{\phi})c_{\phi(i)}.$  That such a path exists is exactly Bach, Luby, and Goldwasser's result!

\begin{restatable}[Bach, Luby, and Goldwasser, cited in Gilmore, Lawler, and Shmoys \cite{Gil85}]{prop}{HP} \label{prop:HP}
Consider a circulant instance on the complete graph $K_n$ where all edges of length $i$ have cost $c_i$, for $i=1, ...., d$.  Let $\phi$ be an associated stripe permutation.  A minimum-cost Hamiltonian path has cost $$\sum_{i=1}^{\ell} (g_{i-1}^{\phi}-g_i^{\phi})c_{\phi(i)}.$$  Such a path uses $(g_{i-1}^{\phi}-g_i^{\phi})$ edges of length $\phi(i)$, for $i=1, ..., d$. 
\end{restatable}

We sketch the proof of Proposition \ref{prop:HP} in Section \ref{sec:lb}, when we extend Hamiltonian paths using $(g_{i-1}^{\phi}-g_i^{\phi})$ edges of length $\phi(i)$ for $i=1, ..., \ell$ to Hamiltonian cycles.  Note that  $(g_{i-1}^{\phi}-g_i^{\phi})=1-1=0$ for $i>\ell$, so it is equivalent to say that such a path uses $(g_{i-1}^{\phi}-g_i^{\phi})$ edges of length $\phi(i)$ for $i=1, ..., \ell$, and 0 otherwise.

\section{The Edge-Length Polytope $EL(n)$}\label{sec:el}

\subsection{Traveling Salesman Problem Polytopes: $STSP(n)$ and $EL(n)$}

Substantial research has gone into studying the \emph{symmetric Traveling Salesman Problem polytope}, which we denote by $STSP(n)$. Let $E_n$ denote the edge set of the complete graph $K_n$ on $n$ vertices. For a Hamiltonian cycle $H$ on $[n]$, let $\chi_H\in \R^{|E_n|}$ denote the edge-incidence vector of $H$ (i.e. $\chi_H$ is indexed by the edges of $K_n$ and, for an edge $e\in E_n$,   $(\chi_H)_e=1$ if $e\in H$ and $(\chi_H)_e=0$ otherwise).  $STSP(n)$ is the convex hull of the edge-incidence vectors of all Hamiltonian cycles on $K_n$.  That is, $$STSP(n)=\text{conv}\{\chi_H : H \text{ is a Hamiltonian cycle on } K_n\}\subset \R^{|E_n|}.$$  

Classic polyhedral theory gives several equivalent characterizations of vertices of polytopes.  For a polytope $P\subset \R^m$, a point $x\in P$ is a \emph{vertex} if it is the unique optimal solution to some linear program where $P$ is the set of feasible points (i.e., there exists some $c\in \R^m$ such that $c^Tx < c^Ty$ for all other $y\in P, y\neq x$).  The point $x$ is an \emph{extreme point} if it cannot be written as a nontrivial convex combination of two other points in $P$ (i.e., there are no points $y, z\in P$ with $y\neq x$ and $z\neq x$, and a scalar $\lambda\in [0, 1]$ such that $x=\lambda y+(1-\lambda)z$).  These are equivalent: $x$ is an extreme point if and only if it is a vertex of $P$.  When optimizing a linear program over a polytope $P$, at least one optimal solution will be a vertex.  A polytope $P$ can be represented as the convex hull of its extreme points, but can also be represented as $\{x\in \R^m: Ax\leq b\}$ for some matrix $A$ and vector $b$.  See, e.g., Chapter 2 of Bertsimas and Tsitsiklis \cite{Ber97}, for more polyhedral background.  

The vertices of $STSP(n)$ can be readily characterized.  Each Hamiltonian cycle $H$ is the unique optimal solution to a linear program over $STSP(n)$, e.g. where $c_e=0$ for $e\in H$ and $c_e=1$ otherwise.  The vertices of $STSP(n)$ correspond to the $\frac{(n-1)!}{2}$ Hamiltonian cycles\footnote{There are $n!$ permutations of $[n]$, and each Hamiltonian cycle corresponds to $2n$ permutations: pick any of the $n$ vertices in the Hamiltonian cycle as the start of the permutation, and then obtain two permutations by proceeding clockwise or counterclockwise through the cycle.} on $n$.    

In providing complete descriptions of $STSP(n)$, early work focused on counting and describing \emph{facets}.  Part I, Chapter 4 of Wolsey and Nemhauser \cite{Wol99} provides classic background on facets, and Chapter 5 of  Applegate, Bixby, Chvatal, and Cook \cite{App06b} surveys the prominent role of facets within TSP research (see also Gr{\"o}tschel and Padberg \cite{Gro86} and Naddef \cite{Nad06}). In the early 90s, Boyd and Cunningham \cite{Boyd91} and Christof,  J{\"u}nger, and Reinelt \cite{ch91} respectively found the number of facets of $STSP(7)$ and $STSP(8)$, and a few years later, Christof and Reinelt \cite{Chr96} found the number of facets of $STSP(9)$.  See Table \ref{tab:PolyCounts}.

\begin{table}[t!]
 \begin{center}
 \caption{Vertex and Facet Counts of $STSP(n)$ and $EL(n)$} \label{tab:PolyCounts}
\begin{tabular}{|c|c|c|c|c|} 
 \hline
& \multicolumn{2}{|c|}{$STSP(n)$}  & \multicolumn{2}{|c|}{$EL(n)$}  \\ \hline
 $n$ & $\#$ Vertices& $\#$  Facets & $\#$  Vertices  & $\#$  Facets\\ \hline
6 & 60  & 100 & 5 &5  \\ \hline
7 &  360 & 3,437 & 3 & 3  \\ \hline
8 &  2520 & 194,187 &  10&  8 \\ \hline
9 & 20,160 & 42,104,442 & 6 & 5  \\ \hline
10 & 181,440 & ? & 18 & 10  \\ \hline
11 &  1,814,400 & ? & 5 & 5  \\ \hline
12 & 19,958,400 & ? & 48 & 16  \\ \hline
13 &239,500,800  & ? & 6 & 6  \\ \hline
14 &3,113,510,400  & ? & 51 & 19  \\ \hline
\end{tabular}
\end{center}  
\end{table}

Solving an instance of the (general) symmetric TSP is equivalent to optimizing a linear function over  $STSP(n)$, but circulant TSP adds substantial symmetry.  Since all edges of a given length $i$ have the same cost $c_{i}$, all that matters is how many edges of each length are used.  Specifically, for $x\in STSP(n)$, let $$t_i = \sum_{\{s, t\}\in E: \ell_{s, t} = i} x_{s, t}$$ denote the total weight of edges of length $i$.  As above, let $d=\lfloor \f{n}{2}\rfloor$ be the longest possible length of an edge in $K_n$.  We describe $(t_1, ..., t_d)$ as the \emph{edge-length vector}.   Gutekunst and Williamson \cite{Gut20} introduced  the   \emph{edge-length polytope} $EL(n)$, which is the projection of $STSP(n)$ to the variables $t_1, ..., t_d.$  That is, $$EL(n):= \text{conv}\left\{(t_1, ..., t_d): x\in STSP(n), t_i = \sum_{\{s, t\}\in E: \ell_{s, t} = i} x_{s, t}\right\}.$$ 

Table \ref{tab:PolyCounts} also shows the vertex and facet counts of $EL(n)$ for small $n$.  These were obtained using Polymake \cite{Gaw00} with a brute-force algorithm: iterate over all Hamiltonian cycles on $[n]$, create a list $L$ of all the unique edge-length vectors, enter a polytope into Polymake as the convex hull of $L$, and then use Polymake to compute the vertices and facets of that polytope.  Figures \ref{fig:EL67} and \ref{fig:EL89} show $EL(n)$ for $n=6$ through 9.  Note that $t_1+t_2+...+t_d=n$, so these figures show the project onto $(t_1, ..., t_{d-1})$.  Finally, Appendix \ref{sec:app} lists all vertices of $EL(7)$ through $EL(14).$

\begin{figure}[t!]
\centering
    \begin{subfigure}[t]{0.48\textwidth}
        \centering
\includegraphics[width=\textwidth]{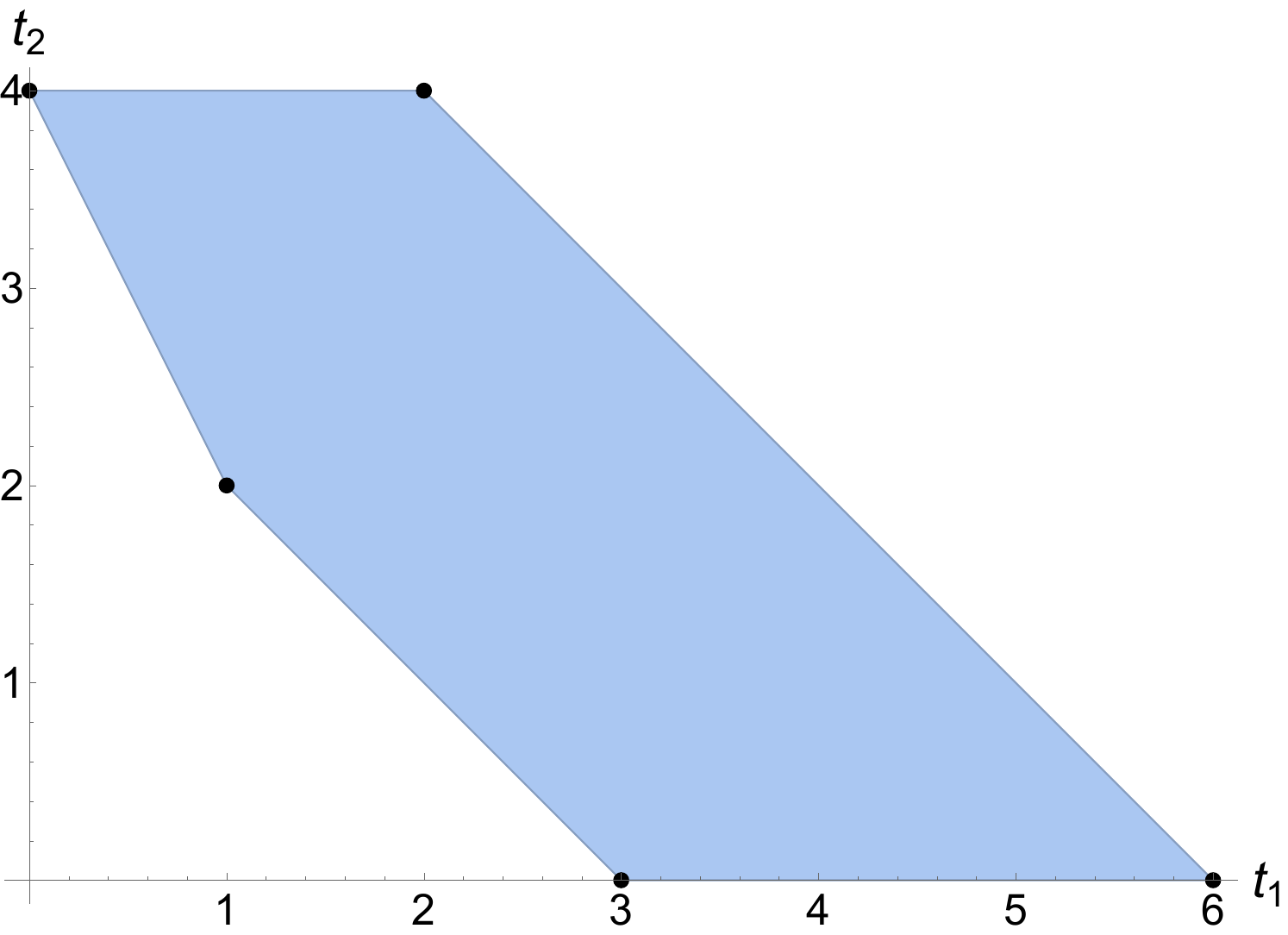}
    \end{subfigure}%
    ~ \hspace{5mm}
    \begin{subfigure}[t]{0.48\textwidth}
        \centering
\includegraphics[width=0.87\textwidth]{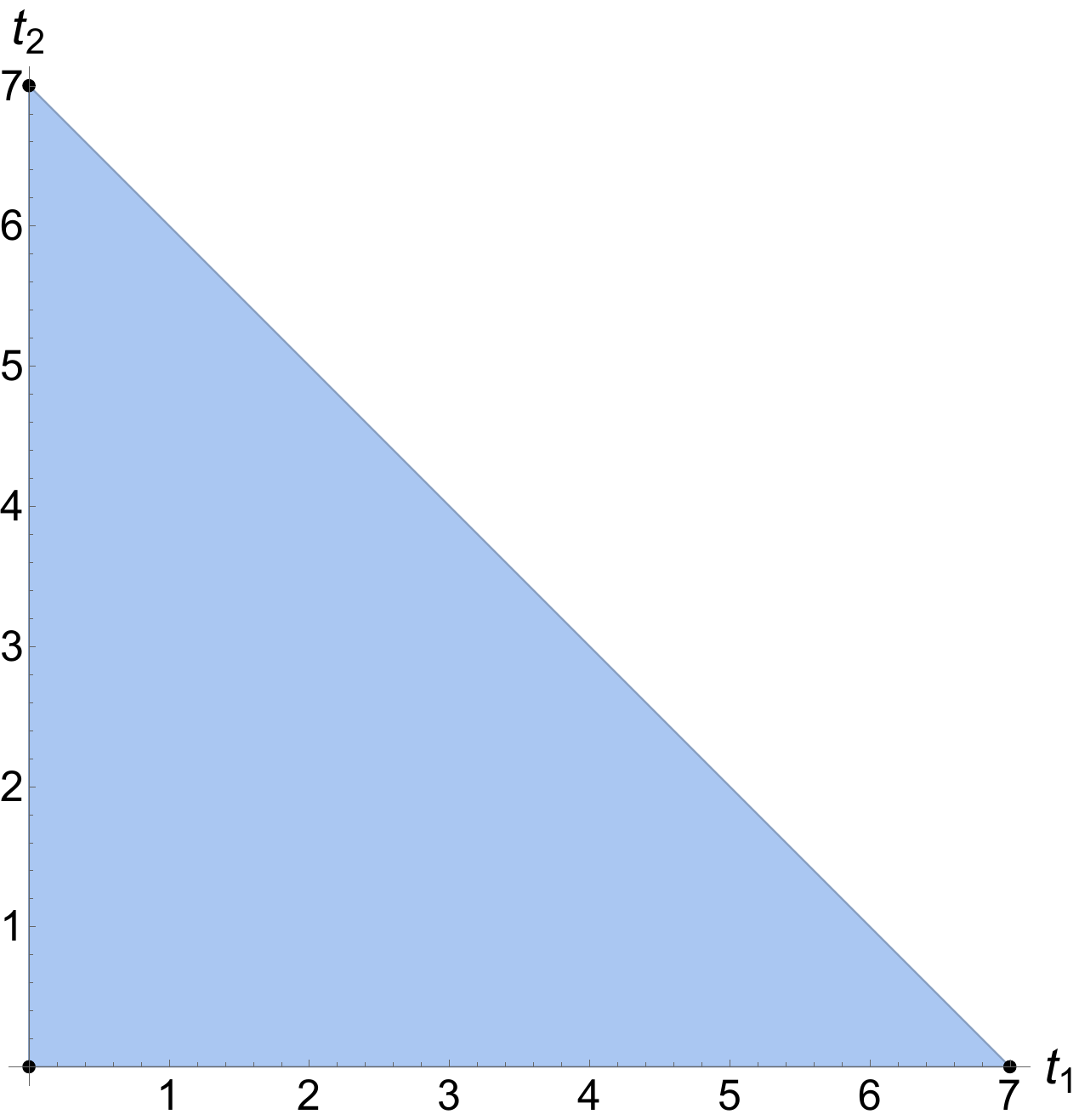}
    \end{subfigure}
    \caption{The Edge-Length Polytope $EL(n)$ for $n=6$ (left) and $n=7$ (right)} \label{fig:EL67}

\end{figure}

\begin{figure}[t!]
\centering
    \begin{subfigure}[t]{0.48\textwidth}
        \centering
\includegraphics[width=0.9\textwidth]{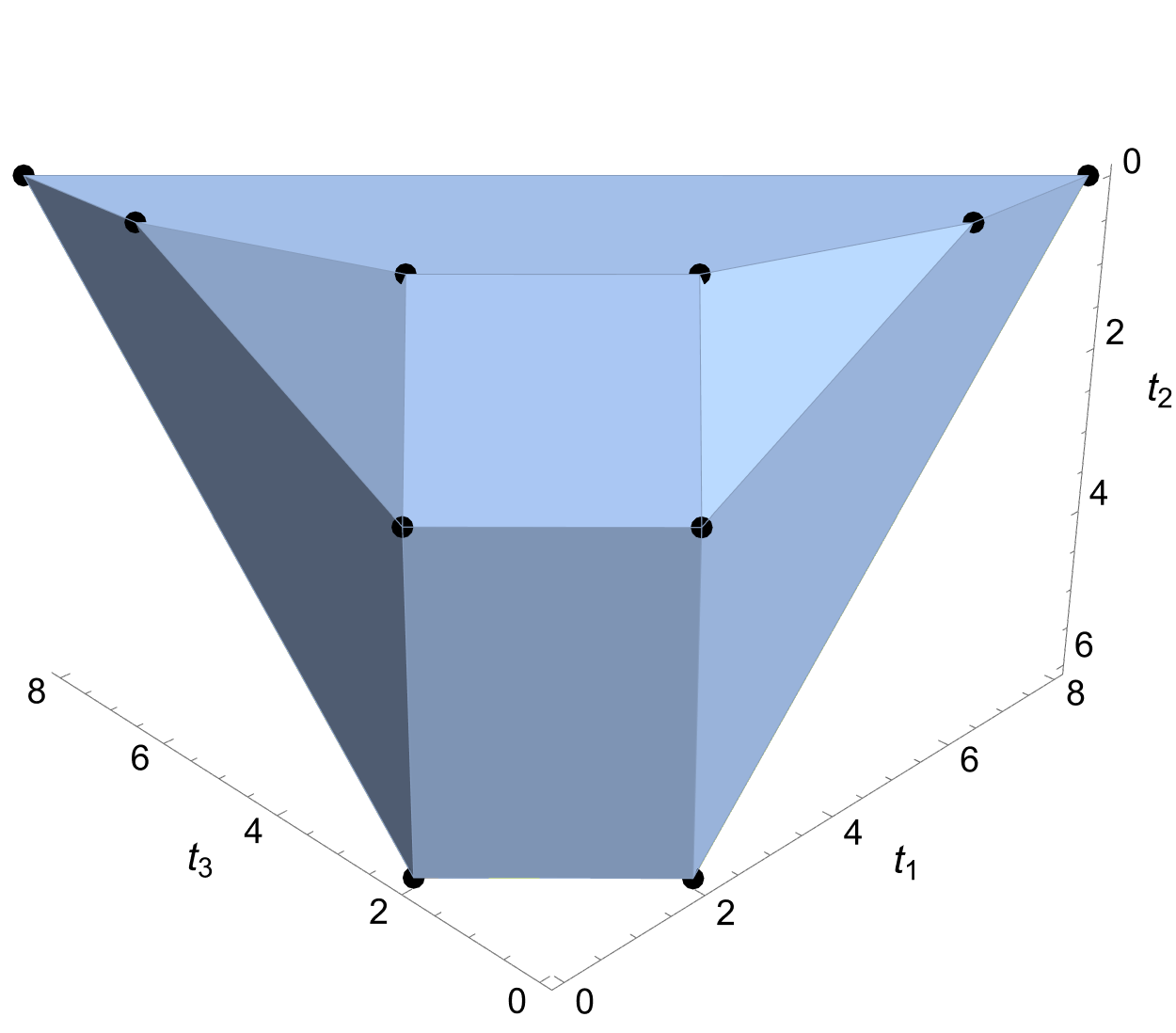}
    \end{subfigure}%
    ~ \hspace{5mm}
    \begin{subfigure}[t]{0.48\textwidth}
        \centering
\includegraphics[width=0.95\textwidth]{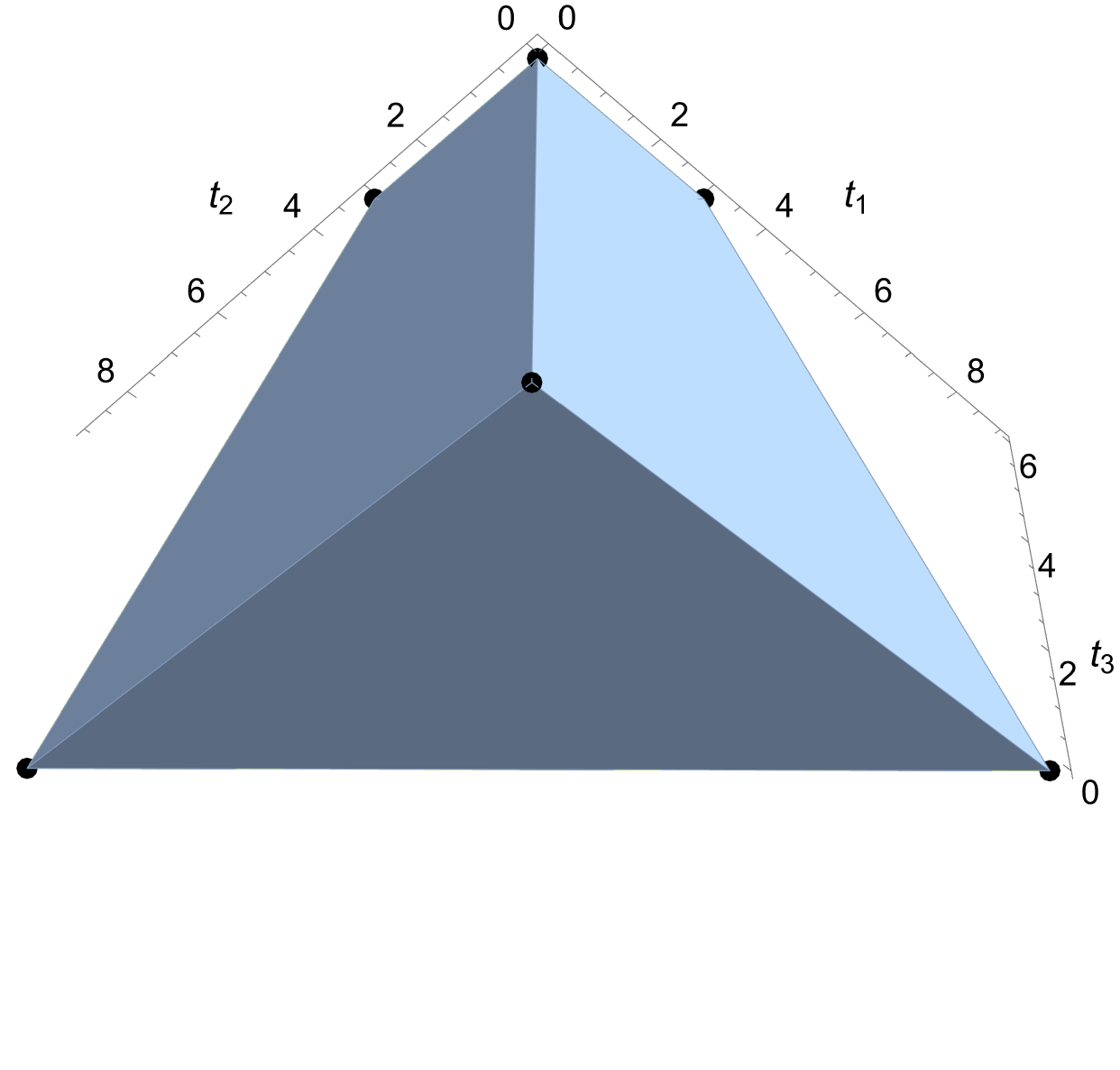}
    \end{subfigure}
    \caption{The Edge-Length Polytope $EL(n)$ for $n=8$ (left) and $n=9$ (right)} \label{fig:EL89}

\end{figure}

One immediate observation from Table \ref{tab:PolyCounts} is that, unlike for $STSP(n)$, the vertex counts of $EL(n)$ do not monotonically increase with $n$. Indeed, they are particularly small whenever $n$ is prime.  This follows from Garfinkel \cite{Gar77}'s classic 70's result that circulant TSP can be easily and efficiently solved whenever the number of vertices $n$ is prime.   The idea is relatively straightforward: Let $\phi(1) \in [d]$ be the length of a cheapest edge.  Then, since $n$ is prime, $\phi(1)$ and $n$ are relatively prime.  Hence, there is a Hamiltonian cycle using exactly $n$ edges of length-$\phi(1)$, attained by starting at vertex $1$ and proceeding to $1+\phi(1)$ then $1+2\cdot \phi(1)$ and so on, until returning to $1+n\cdot \phi(1) \equiv_n 1$ (and this will not repeat any intermediate vertices because $\phi(1)$ and $n$ are relatively prime). Any time there is a unique cheapest edge length $\phi(1)$, this tour will be the unique optimal solution. If there are multiple cheapest edge-lengths, then there is not a unique optimal solution.  When $n$ is prime, there are exactly $d$ vertices of $EL(n)$; these vertices have the form $$(n, 0, ..., 0), (0, n, 0, ..., ), ..., (0, 0, ..., 0, n) \in \R^d.$$  Summarizing:   

\begin{prop}
Let $n$ be prime.  The polytope $EL(n)$ has exactly $d=\lfloor \frac{n}{2}\rfloor$ vertices and  $d=\lfloor \frac{n}{2}\rfloor$ facets.  Specifically, the polytope $EL(n)$ is a scaled simplex $\{x\in \R^n: x_1+...+x_n = n, x_i \geq 0 \text{ for } i=1, ..., d\}.$
\end{prop}

Beal, Bouabida, Gutekunst, and Rustad \cite{Bea24} showed that circulant TSP can also be easily and efficiently solved whenever the number of vertices is of the form $n=p^2$ where $p\geq 3$ is a prime. Let $\phi(1)\in [d]$ be the length of a cheapest edge. If $\phi(1)$ is relatively prime to $n$ (i.e., $\phi(1)$ is not a multiple of $p$), then an optimal solution involves using $n$ edges of length-$\phi(1)$.  Otherwise, let $\phi(i) \in [d]$ be the length of a cheapest edge relatively prime to $n$ (so that $\phi(1), ..., \phi(i-1)$ are all multiples of $p$).  Then Beal, Bouabida, Gutekunst, and Rustad \cite{Bea24} show that there is always an optimal tour using $(n-p)$ edges of length $\phi(1)$ and $p$ edges of length $\phi(i)$.  One of these two types of tours will always be optimal, and it is straightforward to construct tours where they are the unique optimal solution (e.g. for the latter case, where $c_{\phi(1)}=0, c_{\phi(i)}=1,$ and all other edges have cost 2).  This allows us to quickly compute the number of vertices of $EL(n)$ whenever $n$ is a prime-squared:

\begin{prop}
Let $n=p^2$ where $p\geq 3$ is a prime.  The polytope $EL(n)$ has exactly $\frac{p^3-p}{4}$ vertices
\end{prop}

\begin{proof}
    Let $e_i\in \R^d$ denote the vector whose $i$-th coordinate is 1 and where all other coordinates are 0. By Beal, Bouabida, Gutekunst, and Rustad \cite{Bea24}, there are two types of vertices:
    \begin{enumerate}
        \item Vertices corresponding to Hamiltonian cycles using a single edge length, where the length of that edge is relatively prime to $n$: those of the form $n\cdot e_i$ for $1\leq i\leq d$ such that $\gcd(n, i)=1$.
        \item Vertices corresponding to a Hamiltonian cycle using $n-p$ edges whose length is a multiple of $p$, and $p$ edges of a length relatively prime to $p$: those of the form $(n-p) \cdot e_s+p \cdot e_{\ell}$ where $1 \leq s, \ell \leq d$, $\gcd(s, n)=p$, and $\gcd(\ell, n)=1$.
    \end{enumerate}
    To count these vertices, note that there are $d=\lfloor n/2 \rfloor$ edge lengths, and $\lfloor p/2\rfloor $ of these will be not be relatively prime to $n$ (specifically, the edge lengths $p, 2\cdot p, ..., \lfloor p/2 \rfloor \cdot p$).  There are thus $$\left\lfloor \frac{n}{2} \right\rfloor - \left\lfloor \frac{p}{2} \right\rfloor = \frac{n-p}{2}$$ vertices of the first type (using the fact that $n, p$ are odd to simplify floors).  Similarly, we can count vertices of the second type:
    \begin{align*}
      \left(  \left\lfloor \frac{n}{2} \right\rfloor - \left\lfloor \frac{p}{2} \right\rfloor\right)\left\lfloor \frac{p}{2} \right\rfloor &= \left(\frac{n-p}{2}\right) \left(\frac{p-1}{2}\right).
    \end{align*}
    In total, $EL(n)$ thus has
\begin{align*}
    \left(\frac{n-p}{2}\right) +  \left(\frac{n-p}{2}\right) \left(\frac{p-1}{2}\right) & = \frac{n-p}{2} \left(1+\frac{p-1}{2}\right)\\
    &=\frac{(n-p)(p+1)}{4} \\
    &= \frac{n-p^2+np-p}{4} \\
    &= \frac{(n-1)p}{4}
\end{align*} 
    vertices.
    \hfill 
\end{proof}

While the number of vertices of $EL(n)$ grows linearly when $n$ is prime, and like $n^{3/2}$ when $n$ is a prime-squared, our main result is that the number of vertices of $EL(n)$ is not bounded by any polynomial.  In particular, in Section \ref{sec:lb} we prove the following:

\begin{restatable}{thm}{mainthm}\label{mainthm}
Let $n=2^k$ for $k\in \Z, k\geq 4$.  The polytope $EL(n)$ has at least $$2^{k-2} \prod_{i=1}^{k-3}(2^i+1)>2^{\frac{k^2-3k-4}{2}}$$ vertices.
\end{restatable}

Note that  $2^{\frac{k^2-3k-4}{2}}=\frac{1}{4} n^{\f{\log_2(n)-3}{2}},$ so that this result indeed shows a super-polynomial number of vertices.  Hence, a brute-force algorithm for circulant TSP that enumerates and checks all vertices of $EL(n)$ may be efficient for certain cases (like $n$ prime or a prime-squared), but will not in general be efficient.  

To prove Theorem \ref{mainthm}, we will proceed in two main steps.  First, in Section \ref{sec:combo}, we consider a set of Hamiltonian paths closely related to circulant TSP: these are BLG paths, and are essentially the minimum-cost Hamiltonian paths on circulant instances. As a technical note, every Hamiltonian path could be a minimum-cost Hamiltonian path on a circulant instance if, e.g., all edge-lengths have the same cost.  More precisely, BLG paths are thus ``Hamiltonian paths that can arise as the unique optimal solution with some set of circulant costs, using the algorithm of Bach, Luby, and Goldwasser.''  To count BLG paths, we introduce a notion of \emph{encoding sequences} and show that they are in bijection with minimum-cost Hamiltonian paths.  We then count the number of encoding sequences.

Second, in Section \ref{sec:lb}, we show that these Hamiltonian paths can be extended to a Hamiltonian cycle that give rise to vertices of $EL(n)$.  To simplify our arguments and get a super-polynomial lower bound, we take a bit of a cavalier approach and show that we get at least $\frac{1}{4}2^k \prod_{i=1}^{k-3}(2^i+1)$ unique vertices.

\section{Intermezzo: Counting Hamiltonian Paths}\label{sec:combo}
Recall that Proposition \ref{prop:HP} gives a combinatorial formula for the cost/edges of a minimum-cost Hamiltonian paths on a circulant instance: given a cost permutation $\phi$ sorting the edge costs from least to greatest,  and an associated $g$-sequence defined by $$g_i^{\phi}=\begin{cases} n, & i=0 \\ \gcd\left(\phi(i), g_{i-1}^{\phi}\right), & \text{ else,} \end{cases}$$ a minimum-cost path uses $(g_{i-1}^{\phi}-g_i^{\phi})$ edges of length $\phi(i)$, for $i=1, ..., d$.

Throughout this section, we will let $n=2^k$ be a power of 2 (and $d=\lfloor n/2 \rfloor = 2^{k-1}$).  We will also assume, without loss of generality, that $c_{\phi(1)} < c_{\phi(2)}<\cdots < c_{\phi(d)},$ so that $\phi$ uniquely sorts the edge-lengths in increasing cost (and all edge-lengths have unique costs).  This is without loss of generality as, e.g., replacing each $c_{\phi(i)}$ by $c_{\phi(i)}+i$ will yield the same Hamiltonian path according to Proposition \ref{prop:HP}, but will force the edge-length costs to be strictly increasing (when sorted by $\phi$). 

\subsection{Characterizing BLG Paths}

The first step  in our proof of Theorem \ref{mainthm} is to give a combinatorial encoding, \emph{encoding sequences}, for BLG paths; we will then show that this is a superpolynomially-sized set of Hamiltonian paths on $n$ vertices (when $n$ is a power of 2).

First, we consider the structure of $g$-sequences when $n$ is a power of 2:
\begin{cm}\label{cm:small}
    Let $n=2^k$.  Then $n=g_0^{\phi} \geq g_1^{\phi} \geq g_2^{\phi} \geq \cdots \geq g_d^{\phi}$ and each $g_i^{\phi}$ is a power of $2$.
\end{cm}

\begin{proof}
    This claim immediately follows from the definition of a $g$-sequence: $g_i^{\phi}=\gcd(\phi(i), g_{i-1}^{\phi})$, so $g_i^{\phi}$ divides $g_{i-1}^{\phi}$ (which, in turn, divides $g_{i-2}^{\phi}$, and so on).  Hence,  each $g_i^{\phi}$ divides $g_0^{\phi}=n=2^k$, which means that  each $g_i^{\phi}$ is a power of 2, and the $g$ sequence is weakly decreasing.
\end{proof}

Similarly, a minimum-cost Hamiltonian path only uses edges of length $\phi(i)$ if $g_i^{\phi} < g_{i-1}^{\phi}$.  This implies that a minimum-cost Hamiltonian path on $n=2^k$ vertices can only use at most $k$ different edge-lengths: one for each time where the $g$ sequence can decrease. In particular:
\begin{prop}\label{prop:lim}
    Let $n=2^k.$  Partition the possible edge-lengths $[d]$ into sets $S_1, S_2, ..., S_{k}$ where $$S_i=\{1\leq j\leq d: \gcd(n, j)=2^{k-i}\}, i=1, ..., k.$$  Then a minimum-cost Hamiltonian path will use at most one edge-length from each $S_i$. Specifically, let $i_1, i_2\in S_i$.  If $\phi^{-1}(i_2)>\phi^{-1}(i_1)$ (so that $i_2$ comes after $i_1$ in $\phi$, or equivalently, $c_{i_2}>c_{i_1}$), then a minimum-cost Hamiltonian path will not use any edges of length $i_2$.  
\end{prop}
For instance, when $n=2^4=16$, this proposition says that a Hamiltonian path can use at most one edge from each of the following sets: $S_1=\{8\}, S_2=\{4\}, S_3=\{2, 6\}, S_4=\{1, 3, 5, 7\}$.  These are, respectively, the sets of edge-lengths whose gcd with $n$ are 8, 4, 2, and 1.  Moreover, it can only use the cheapest edge-length in each set. 

\begin{proof}
    Suppose that we have two edge-lengths $i_1, i_2\in S_i$ whose gcd with $n$ is $2^{k-i}$, and where $\phi^{-1}(i_2)>\phi^{-1}(i_1)$ (i.e. $i_2$ is a more expensive edge).  Then we can write $i_1=k_1\cdot 2^{k-i}$ and  $i_2=k_2\cdot 2^{k-i}$, where both $k_1$ and $k_2$ are odd integers.  To simplify notation, we suppress the superscript $\phi$ on $g$ (i.e. writing $g_{\phi^{-1}(i_2)}$ instead of $g_{\phi^{-1}(i_2)}^{\phi}$).

    By Claim \ref{cm:small}, $g_{\phi^{-1}(i_1)}$ is a power of 2.  Moreover,  $g_{\phi^{-1}(i_1)}$  definitionally divides  $\phi(\phi^{-1}(i_1))=i_1=k_1 \cdot 2^{k-i}$ with $k_1$ odd.  Together, these imply that $g_{\phi^{-1}(i_1)}\leq 2^{k-i}$.  

    Also by Claim \ref{cm:small}, $g$-sequences are (weakly) decreasing.  Thus, since $g_{\phi^{-1}(i_1)}\leq 2^{k-i}$, we also have that $g_{\phi^{-1}(i_1)+1}\leq 2^{k-i}, g_{\phi^{-1}(i_1)+2}\leq 2^{k-i}, ..., g_{d}\leq 2^{k-i}$.  Since  $\phi^{-1}(i_2)>\phi^{-1}(i_1)$, this means $g_{\phi^{-1}(i_2)-1}\leq 2^{k-i}$ (and is again a power of 2).  Hence, we can write $g_{\phi^{-1}(i_2)-1}= 2^{j}$ for some $j\leq k-i$.  Finally, by definition,
    \begin{align*}
        g_{\phi^{-1}(i_2)} &= \gcd(\phi(\phi^{-1}(i_2)), g_{\phi^{-1}(i_2)-1}) \\
        &= \gcd(i_2, g_{\phi^{-1}(i_2)-1})  \\
        & = \gcd(k_2 \cdot 2^{k-i}, 2^j) \\
        &= 2^j \tag{since $j\leq k-i$ and $k_2$ is odd} \\
      &=  g_{\phi^{-1}(i_2)-1}.
    \end{align*}
Hence, $g_{\phi^{-1}(i_2)}=  g_{\phi^{-1}(i_2)-1}$ and thus Proposition \ref{prop:HP} indicates that there will be 0 edges of length $i_2$. \hfill
\end{proof}

Note that Proposition \ref{prop:lim} does not imply that a BLG path necessarily uses an edge-length from every $S_i$.  Instead, it says that the only possible edge-length a BLG path will use from $S_i$ is the cheapest.  However, we can readily characterize exactly when a BLG path uses an edge-length from $S_i$: it only uses the cheapest edge-length from $S_i$ if it is cheaper than all edge-lengths in $S_{i+1}, S_{i+2}, ....$, and $S_{k}$.  For instance, if $n=16$, we can only use an edge-length from $S_2=\{4\}$ if it is cheaper than all edge lengths in $S_3\cup S_4= \{1, 2, 3, 5, 6, 7\}$ (if any edge-length from $S_3\cup S_4$ was cheaper than 4, then the $g$-sequence would have already decreased to 1 or 2 by the time it reached $g_{\phi^{-1}(4)}^{\phi}$).

\begin{prop}\label{prop:cheap}
    Let $n=2^k.$  Partition the possible edge-lengths $[d]$ into sets $S_1, S_2, ..., S_{k}$ where $$S_i=\{1\leq j\leq d: \gcd(n, j)=2^{k-i}\}, i=1, ..., k.$$  Let $i^*$ be the cheapest edge-length in $S_i$.  A BLG path will use edges of length $i^*$ if and only if $$c_{i^*}< c_j \text{ for all } j \in \left( S_{i+1}\cup S_{i+2} \cup \cdots \cup S_{k}\right).$$ 
\end{prop}

\begin{proof}
    The proof of this proposition proceeds relatively directly.   To simplify notation, we again suppress the superscript $\phi$ on $g$.  
    
    Suppose there is some edge-length $j^*\in S_j$, with $j\geq i+1$, such that $c_{j^*}<c_{i^*}$.  Then $\phi^{-1}(j^*) < \phi^{-1}(i^*)$.  The main idea is that, since $j^*$ appears before $i^*$ in the cost permutation and has fewer factors of $2$, the term $g_{\phi^{-1}(i^*)-1}$ will already be too small for the edge-length $i^*$ to effect the $g$-sequence.  More formally, $$g_{\phi^{-1}(j^*)} = \gcd(\phi(\phi^{-1}(j^*)),g_{\phi^{-1}(j^*)-1}) =\gcd(j^*,g_{\phi^{-1}(j^*)-1})  \leq 2^{k-j},$$ since $j^*\in S_j$ implies that $j^* = k_j \cdot 2^{k-j}$ for some odd $k_j$ and, by Claim \ref{cm:small}, all elements of the $g$ sequence are powers of 2.    Thus $$g_{\phi^{-1}(j^*)}  \leq 2^{k-j} < 2^{k-i},$$ since $j\geq i+1.$    
    Since $\phi^{-1}(j^*) < \phi^{-1}(i^*)$, an analogous argument to the proof of Proposition \ref{prop:lim} shows that $g_{\phi^{-1}(i^*)}= g_{\phi^{-1}(i^*)-1}$ and thus the BLG path will not use any edges of length $i^*$. 

    In contrast, suppose that no such edge $j^*$ exists.  Then $\phi(1), \phi(2), ..., \phi(\phi^{-1}(i^*)-1)$ are all in $S_1 \cup S_2 \cup \cdots \cup S_{i-1}$ and so divisible by $2^{k-i+1}$.  Hence $g_{\phi^{-1}(i^*)-1}\geq 2^{k-i+1}$, but $$g_{\phi^{-1}(i^*)} = \gcd(\phi(\phi^{-1}(i^*)),g_{\phi^{-1}(i^*)-1}) =\gcd(i^*,g_{\phi^{-1}(i^*)-1}) =2^{k-i}.$$  Thus the BLG path will use  $$(g_{\phi^{-1}(i^*)-1}-g_{\phi^{-1}(i^*)}) \geq 2^{k-i+1}-2^{k-i} \geq 2^{k-i} \geq 1$$ edges of length $i^*.$  \hfill
\end{proof}
As a brief remark, note that we must use exactly one edge-length from $S_{k}=\{1, 3, 5, ..., d-1\}$, whichever is cheapest; once we use this edge-length, all remaining terms in the $g$ sequence will be 1 (and the $g$ sequence will not reach 1 until reaching that edge-length).

Proposition \ref{prop:cheap} gives rise to a combinatorial approach for counting BLG paths.  Consider, for example, the case when $n=16:$
\begin{enumerate}
    \item We can either use an edge-length from $S_1=\{8\}$ or skip $S_1$ (and use edges of length 8 if and only if they are the cheapest edge).
    \item We can either use an edge-length from $S_2=\{4\}$ or skip $S_2$ (and use edges of length 4 if and only if they are cheaper than edges of lengths 1, 2, 3, 5, 6, and 7).
     \item We can either use an edge-length from $S_3=\{2, 6\}$ or skip $S_3$ (and we can only use whichever of those two edge-lengths are cheaper, and we use that edge if and only if it is cheaper than edges of lengths 1, 3, 5, 7).
     \item We use exactly one edge-length from $S_4=\{1, 3, 5, 7\}$, whichever is cheapest.
\end{enumerate}

We can succinctly describe this decision process as an \emph{encoding sequence}: a sequence $(s_1, ..., s_k)$ that proceeds from $S_1$ to $S_{k}$, where $s_i$ is either an element of $S_i$ or indicates that we skip $S_i$ (with the restriction that we cannot skip $S_k$).  More formally:
\begin{defn}
Let $n=2^k$ for $k\geq 3$.  An {\bf encoding sequence} is a sequence $(s_1, ..., s_k)$ where, for $i=1, ..., k-1$, $$s_i\in \{1\leq j\leq d: \gcd(n, j)=2^{k-i}\} \cup \{x\}$$ and $$s_k \in \{1\leq j\leq d: \gcd(n, d)=1\}.$$
\end{defn}
Setting $s_i=x$ in the encoding sequence means that we ``skip'' $S_i$ and there is no weight on any edge-length in $S_i$. Otherwise $s_i$ is an edge-length that for which the BLG path will use a strictly-positive number of edges (i.e. in the corresponding edge-length vector $t$, we have that $t_{s_i}>0$); by Proposition \ref{prop:lim}, $s_i$ is the only length in $S_i$ that will be used.

Table \ref{tab:BLG8} explicitly shows the edge-lengths and encoding sequences representing all BLG paths on $n=8$ vertices.  

\begin{table}[t!]
\begin{center}
  \caption{The 8 possible BLG paths on $n=8$ vertices.  The first row handles the case where the cheapest edge-length is 1, the second row handles the case where the cheapest edge-length is 3, the third and fourth rows handle the case where it is 2, and the last four rows handle the case where it is 4.  Extended edge-length vectors correspond to Hamiltonian cycles.} \label{tab:BLG8}

    \begin{tabular}{|c|c|c|c|} \hline
         Encoding Sequence & Edge-Length Vector  & Condition & Extended Edge-Length Vector \\ \hline
         $(x, x, 1)$ & $(7, 0, 0, 0)$ & $c_1 < c_2, c_3, c_4$ & $(8, 0, 0, 0)$  \\ \hline
         $(x, x, 3)$ & $(0, 0, 7, 0)$ & $c_3 < c_1, c_2, c_4$ & $(0, 0, 8, 0)$ \\ \hline
         $(x, 2, 1)$ & $(1, 6, 0, 0)$ & $c_2 < c_1 < c_3, c_2 < c_4$  & $(2, 6, 0, 0)$\\ \hline     
         $(x, 2, 3)$ & $(0, 6, 1, 0)$ & $c_2 < c_3 < c_1, c_2 <  c_4$  & $(0, 6, 2, 0)$\\ \hline   
         $(4, x, 1)$ & $(3, 0, 0, 4)$ & $c_4<c_1 < c_2, c_3$ & $(4, 0, 0, 4)$ \\ \hline     
         $(4, x, 3)$ & $(0, 0, 3, 4)$ & $c_4<c_3 < c_1, c_2$  & $(0, 0, 4, 4)$ \\ \hline   
         $(4, 2, 1)$ & $(1, 2, 0, 4)$ & $c_4<c_2<c_1<c_3$  & $(2, 2, 0, 4)$\\ \hline     
         $(4, 2, 3)$ & $(0, 2, 1, 4)$ & $c_4<c_2<c_3<c_1$ & $(0, 2, 2, 4)$ \\ \hline  
    \end{tabular}   
\end{center}
\end{table}

   Any BLG path can be represented via an encoding sequence: Trace through the corresponding cost permutation $\phi$ computing the number of edges of length $\phi(1)$ used ($g_0^{\phi}-g_1^{\phi})$, then the number of edges of length $\phi(2)$ used ($g_1^{\phi}-g_2^{\phi})$, and so on.  Whenever reaching an edge-length that is used (i.e., that causes the $g$-sequence to drop), log that edge length.  By Propositions \ref{prop:lim} and \ref{prop:cheap}, there will be at most one edge-length from each $S_i$, and the edge-lengths used will come from successive $S_i$.

We also observe that the number of edges of a given length used in a BLG path can also be uniquely determined by the encoding sequence: Suppose $s_i\neq x$, so that edges of length $s_i$ are used in the BLG path (and $s_i$ is the cheapest edge-length in $S_i$).  We can compute $t_{s_i}$, the number of edges of length $s_i$ used, based on two cases:
\begin{itemize}
    \item If  $s_j=x$ for all $j<i$ (i.e. if $s_i$ is the cheapest edge-length), then $g_{\phi^{-1}(s_i))-1}=n$ and $g_{\phi^{-1}(s_i))}=2^{k-i}$ .  Hence we will use $g_{\phi^{-1}(s_i))-1}-g_{\phi^{-1}(s_i))}=n-2^{k-i}$ edges of length $s_i$.
    \item Otherwise let $j= \max \{1\leq k < i: s_k \neq x\}$.  Then $g_{\phi^{-1}(s_i))-1}=2^{k-j}$ and $g_{\phi^{-1}(s_i))}=2^{k-i}$, so we use $2^{k-j}-2^{k-i}$ edges of length $s_i$.
\end{itemize}
For any $j\not\in s$, $t_j=0$.  

Since the edge-length vector is uniquely determined by the encoding sequence $s$, if two BLG paths use the same set of edge-lengths, they must be the same BLG path (i.e., have the same edge-length vector):
\begin{cm}\label{cm:unique}
    Suppose that $P$ and $H$  are BLG paths on $n=2^k$ vertices, with respective edge-length vectors $(p_1, ..., p_d)$ and $(h_1, ..., h_d)$.  If $\{i: p_i>0\} = \{i: h_i>0\}$ then $(p_1, ..., p_d)=(h_1, ..., h_d)$
\end{cm}

We can now prove that encoding sequences exactly capture BLG paths:

\begin{prop}\label{prop:bij}
  When $n=2^k$ is a power of $2$, BLG paths are in bijection with encoding sequences.
\end{prop}

\begin{proof}
We have previously described how to represent any BLG path as an encoding sequence.  This map is a bijection.  First, suppose $P_1$ and $P_2$ are distinct BLG paths (i.e., they have different edge-length vectors).  By Claim \ref{cm:unique},  there must be some edge-length used in one but not the other.  Without loss of generality, suppose that edge length $i$ is used in $P_1$ but not $P_2$.  Then $i$ will appear in the encoding sequence of $P_1$ but not $P_2$, so that they have distinct encoding sequences.  
   
   Similarly, this map is also surjective.  Consider an encoding sequence  $(s_1, ..., s_k)$.  One can construct a circulant instance whose corresponding BLG path maps to that encoding sequence as follows: for $1\leq i\leq k,$ if $s_i\neq x$, set $c_{s_i}=i$.   Then set all other edge costs to be larger than $k$ (in any order).  Then $\phi(1)$ will be the first non-$x$ (``non-skipped'') element in the encoding sequence, $\phi(2)$ will be the second non-$x$ element in the encoding sequence, and so on.  Thus the BLG path corresponding to $\phi$ will exactly have encoding sequence $(s_1, ..., s_k)$.  \hfill
\end{proof}

\subsection{Counting BLG Paths}

Viewing BLG paths as encoding sequences allows us to briskly count them, obtaining our lower bound on A352568 and A030077.  For instance, we previously saw that, when $n=16$,  $S_1=\{8\},  S_2=\{4\}, S_3=\{2, 6\},$ and $S_4=\{1, 3, 5, 7\}.$  An encoding sequence $(s_1, s_2, s_3, s_4)$ will have $s_i\in (S_i \cup \{x\})$ for $i=1, 2, 3$ and $s_4\in S_4$.  Hence there are $$|S_1 \cup \{x\}| \cdot |S_2 \cup \{x\}|  \cdot |S_3 \cup \{x\}|  \cdot |S_4| = 2\cdot 2\cdot 3 \cdot 4 = 48$$ encoding sequences (and thus BLG paths) on $16$ vertices.  More generally:

\begin{prop}\label{prop:BLGPaths}
Let $n=2^k$ for $k\in \Z, k\geq 4$.  There are $2^k \prod_{i=1}^{k-3}(2^i+1)$ BLG paths.
\end{prop}

\begin{proof}
Proposition \ref{prop:bij} shows that encoding sequences are in bijection with BLG paths.  As such, it suffices to count encoding sequences.  We thus count the number of options for each $s_i$.  
\begin{itemize}
    \item For $s_1\in (S_1\cup \{x\})$, $$\{1\leq j\leq d: \gcd(n, j)=2^{k-1}\} =\{1\leq j\leq d: gcd(n, j)=n/2\}  =\{n/2\},$$ so there are two options $(n/2$ or $x$).
    \item For $s_2 \in (S_2\cup\{x\})$, $$\{1\leq j\leq d: \gcd(n, j)=2^{k-2}\} =\{1\leq j\leq d: gcd(n, j)=n/4\}  = \{n/4\},$$ so there are two options ($n/4$ or $x$).
    \item For $s_k\in S_k,$ (which cannot be skipped), $$|1\leq j\leq d: \gcd(n, j)=1\}| = |\{1\leq j\leq d: j \text{ odd}\}| = n/4,$$ so there are $n/4=2^{k-2}$ options
    \item For all other $i=3, ..., k-1$, $s_i\in (S_i \cup \{x\})$, so we consider $\{1\leq j\leq d: \gcd(n, j)=2^{k-i}\}.$  Note that these are exactly the $j\in [d]$ that are divisible by $2^{k-i}$ but not $2^{k-i+1}$.  That is, $$\{1\leq j\leq d: \gcd(n, j)=2^{k-i}\} = \{z\cdot 2^{k-i}: z \text{ odd}, z\cdot 2^{k-i}\leq 2^{k-1}\} = \{z\cdot 2^{k-i}: z \text{ odd}, z\leq 2^{i-1}\}.$$  There are $2^{i-2}$ odd values of $z$ between 1 and $2^{i-1}.$  Including $x$, there are thus $2^{i-2}+1$ choices for $s_i,$ including $x$.
\end{itemize}

Taken together (and treating the options for $s_1, s_2,$ and $s_k$ separately), the number of encoding sequences is:

$$|S_1| \cdot |S_2| \cdot |S_k| \cdot \prod_{i=3}^{k-1} |S_i| = 2\cdot 2\cdot 2^{k-2} \cdot \prod_{i=3}^{k-1} \left(2^{i-2}+1\right) = 2^k \prod_{i=1}^{k-3}  \left(2^{i}+1\right),$$ as claimed. \hfill\end{proof}

Finally, to succinctly show that our bound is super-polynomial, we lower-bound this product:

\begin{prop}
 $$2^k \prod_{i=1}^{k-3}(2^i+1)>2^{\frac{k^2-3k}{2}}.$$
 \end{prop}

\begin{proof}
\begin{align*}
    2^k \prod_{i=1}^{k-3}(2^i+1) &> 2^k \prod_{i=1}^{k-3} 2^i \\
    &= 2^k \cdot 2^{\sum_{i=1}^{k-3} i} \\
    &= 2^k \cdot 2^{\frac{(k-3)(k-2)}{2}} \\
    &= 2^k \cdot 2^{\f{k^2-5k+6}{2}} \\
    &>  2^k  \cdot 2^{\f{k^2-5k}{2}} \\
        &=   2^{\f{k^2-3k}{2}}.
\end{align*}
\hfill
\end{proof}

Recall that OEIS sequence A352568 counts the number of combinations of edge lengths that can comprise a Hamiltonian path.  Since BLG paths are all Hamiltonian paths with distinct edge-length vectors, the above result immediately gives a lower bound on A352568 (whenever $n=2^k$) and shows that A352568 grows super-polynomially.  Recall also that OEIS sequence A030077 is similar, and counts the distinct total lengths of such paths when the vertices $1$ through $n$ are uniformly arranged around a circle (formally, ``Take n equally spaced points on circle, connect them by a path with n-1 line segments; sequence gives number of distinct path lengths'' \cite{Slo}).  McKay and Peters \cite{Mc22} proved that A352568 and A030077 match whenever $n$ is a power of 2, so we also have a lower bound on A030077.

\begin{cor}
    Let $n=2^k$.  Then the $n$th term in OEIS sequences A352568 and A030077 is at least  $$2^k \prod_{i=1}^{k-3}(2^i+1)>2^{\frac{k^2-3k}{2}},$$ and both sequences are super-polynomial. 
\end{cor}

\section{Lower-Bounding the Vertices of $EL(n)$ for Powers of 2}\label{sec:lb}

Let $n=2^k$ be a power of 2.  So far, we have shown that there are at least $2^{\f{k^2-3k}{2}}$ combinations of edge lengths that correspond to Hamiltonian paths.  To prove Theorem \ref{mainthm}, however, we need edge-length vectors of Hamiltonian cycles that correspond to vertices of $EL(n)$.  In this section, we thus construct vertices of $EL(n)$ using edge-length vectors corresponding to the above encoding sequences.  To succinctly get a super-polynomial lower bound, we will argue that at least $\frac{1}{4}$ of the above encoding sequences give rise to a unique vertex of $EL(n)$, which completes the proof of Theorem \ref{mainthm}.  

Encoding sequences correspond to BLG paths.  To extend these paths to cycles, we sketch the proof of Proposition \ref{prop:HP} and how the BLG paths are constructed.   Recall that $\ell=\min\{i:1\leq i\leq d, g_i^{\phi}=1\}$ and $d=\lfloor n/2\rfloor.$

\HP*

 The idea is recursive: start by constructing a path visiting all vertices $\{v\in[n]: v\equiv_{g_1^{\phi}} 1\}.$  Then, for $j=2$ to $\ell$, use circulant symmetry to extend a path visiting all vertices in  $\{v\in[n]: v\equiv_{g_{j-1}^{\phi}} 1\}$ to  a path visiting all vertices in  $\{v\in[n]: v\equiv_{g_{j}^{\phi}} 1\}.$  
 
 While the idea of the proof is not that complicated, it is notationally involved and perhaps easiest to see through an example. Figure \ref{fig:BLGEx1} shows the process when $n=32$ and $\phi(1)=8, \phi(2)=10, \phi(3)=\phi(\ell)=7$, so that $g_1^{\phi}=8, g_2^{\phi}=2,$ and $g_3^{\phi}=1$.  The process starts with a path visiting all vertices in  $\{v \in [32]: v\equiv_8 1\} = \{1, 9, 17, 25\}.$  This path is immediate: we start at vertex 1 and follow edges of length $\phi(1)=8$ until reaching all vertices in $ \{1, 9, 17, 25\}$.  To extend this to a path on all vertices in $\{v \in [24]: v\equiv_{g_2^{\phi}} 1\}=\{1, 3, 5, ..., 31\}$, we first use circulant symmetry to ``copy and translate'' the preceding path by multiples of $\phi(2)=6$.  Then we add edges of length $\phi(2)$ edges to connect the translated paths and attain the path on $\{1, 3, 5, ..., 31\}$.  Finally, the process is repeated, copying and translating this path by $\phi(3)=7$ and then using a length-$\phi(3)$ edge to connect the translate and attain a path on $\{v\in [32]: v\equiv_1 1\}=[n]$, the full vertex set.

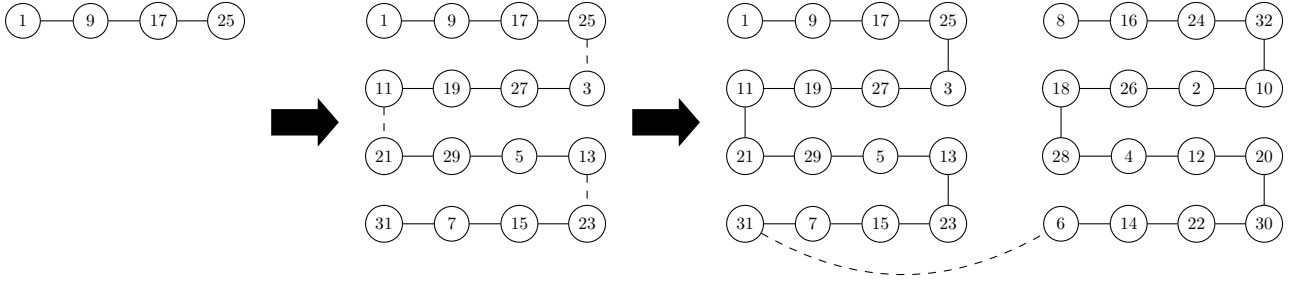
\begin{figure}
    \centering
 \begin{tikzpicture}[scale=0.6, transform shape]

        \tikzset{vertex/.style = {shape=circle,draw,minimum size=2em}}
        \vertex (a0) at (0, 0) {1};
        \vertex (a1) at (1.5, 0) {9};
        \vertex (a2) at (3, 0) {17};
                \vertex (a3) at (4.5, 0) {25};
        \draw (a0) to (a1) to (a2) to (a3);
        
        \vertex (b0) at (0, -1.5) {11};
        \vertex (b1) at (1.5, -1.5) {19};
        \vertex (b2) at (3, -1.5) {27};
                \vertex (b3) at (4.5, -1.5) {3};
        \draw (b0) to (b1) to (b2) to (b3);
        
        \vertex (c0) at (0, -3) {21};
        \vertex (c1) at (1.5, -3) {29};
        \vertex (c2) at (3, -3) {5};
                \vertex (c3) at (4.5, -3) {13};
        \draw (c0) to (c1) to (c2) to (c3);  
        
        \vertex (d0) at (0, -4.5) {31};
        \vertex (d1) at (1.5, -4.5) {7};
        \vertex (d2) at (3, -4.5) {15};
                \vertex (d3) at (4.5, -4.5) {23};
        \draw (d0) to (d1) to (d2) to (d3);  
        
        \draw[dashed] (a3) to (b3);
                \draw[dashed] (b0) to (c0);
                        \draw[dashed] (c3) to (d3);

        \vertex (e0) at (-8, 0) {1};
        \vertex (e1) at (-6.5, 0) {9};
        \vertex (e2) at (-5, 0) {17};
                        \vertex (e3) at (-3.5, 0) {25};
        \draw (e0) to (e1) to (e2) to (e3);

        \draw[-{Triangle[width=18pt,length=8pt]}, line width=10pt](-2.5,-2.25) -- (-1, -2.25);

\draw[-{Triangle[width=18pt,length=8pt]}, line width=10pt](5.5,-2.25) -- (7, -2.25);

        \vertex (a0b) at (8, 0) {1};
        \vertex (a1b) at (9.5, 0) {9};
        \vertex (a2b) at (11, 0) {17};
        \vertex (a3b) at (12.5, 0) {25};
        \draw (a0b) to (a1b) to (a2b) to (a3b);
        
        \vertex (b0b) at (8, -1.5) {11};
        \vertex (b1b) at (9.5, -1.5) {19};
        \vertex (b2b) at (11, -1.5) {27};
        \vertex (b3b) at (12.5, -1.5) {3};
        \draw (b0b) to (b1b) to (b2b) to (b3b);
        
        \vertex (c0b) at (8, -3) {21};
        \vertex (c1b) at (9.5, -3) {29};
        \vertex (c2b) at (11, -3) {5};
        \vertex (c3b) at (12.5, -3) {13};
        \draw (c0b) to (c1b) to (c2b) to (c3b);  
        
        \vertex (d0b) at (8, -4.5) {31};
        \vertex (d1b) at (9.5, -4.5) {7};
        \vertex (d2b) at (11, -4.5) {15};
        \vertex (d3b) at (12.5, -4.5) {23};
        \draw (d0b) to (d1b) to (d2b) to (d3b);  
        
        \draw (a3b) to (b3b);
        \draw (b0b) to (c0b);
        \draw(c3b) to (d3b);

 \vertex (a0c) at (15, 0) {8};
        \vertex (a1c) at (16.5, 0) {16};
        \vertex (a2c) at (18, 0) {24};
        \vertex (a3c) at (19.5, 0) {32};
        \draw (a0c) to (a1c) to (a2c) to (a3c);
        
        \vertex (b0c) at (15, -1.5) {18};
        \vertex (b1c) at (16.5, -1.5) {26};
        \vertex (b2c) at (18, -1.5) {2};
        \vertex (b3c) at (19.5, -1.5) {10};
        \draw (b0c) to (b1c) to (b2c) to (b3c);
        
        \vertex (c0c) at (15, -3) {28};
        \vertex (c1c) at (16.5, -3) {4};
        \vertex (c2c) at (18, -3) {12};
        \vertex (c3c) at (19.5, -3) {20};
        \draw (c0c) to (c1c) to (c2c) to (c3c);  
        
        \vertex (d0c) at (15, -4.5) {6};
        \vertex (d1c) at (16.5, -4.5) {14};
        \vertex (d2c) at (18, -4.5) {22};
        \vertex (d3c) at (19.5, -4.5) {30};
        \draw (d0c) to (d1c) to (d2c) to (d3c);  
        
        \draw (a3c) to (b3c);
        \draw (b0c) to (c0c);
        \draw(c3c) to (d3c);
                    
                        \draw[dashed, bend right] (d0b) to (d0c);

        \end{tikzpicture}
    \caption{Producing a BLG path on an input on $n=32$ vertices with $\phi(1)=8, \phi(2)=10, \phi(3)=\phi(\ell)=7$. A path is extend from  $\{v\in [32]: v\equiv_{8} 1\}$ to  $\{v\in [32]: v\equiv_{2} 1\}$ to to  $\{v\in [32]: v\equiv_{1} 1\}$.  Each step uses circulant symmetry to  ``copying and translate'' the previous path, and then connect translates.  At each step, edges connecting translates are dashed.}
    \label{fig:BLGEx1}
\end{figure}

\begin{proof}[Proof (Sketch)]
First, we show how to construct a Hamiltonian path with $(g_{i-1}^{\phi}-g_i^{\phi})$ edges of length $\phi(i)$, for $i=1, ..., \ell$, using circulant symmetry.  We begin by constructing a path visiting all vertices in $\{v\in[n]: v\equiv_{g_1^{\phi}} 1\}.$  Then, for $j=2$ to $\ell$, we use circulant symmetry to extend our path visiting all vertices in  $\{v\in[n]: v\equiv_{g_{j-1}^{\phi}} 1\}$ to  a path visiting all vertices in  $\{v\in[n]: v\equiv_{g_{j}^{\phi}} 1\}.$

In general, we begin with a path $\{v\in[n]: v\equiv_{g_1^{\phi}} 1\}$ attained by starting at vertex 1 and following edges of length-$\phi(1)$ until reaching the last vertex in that set (using $n/g_1^{\phi}-1$ edges; one more would create a cycle).  Then, for $j=2$ to $\ell$, suppose $v_1, v_2, ..., v_{n/g_{j-1}^{\phi}}$ is a path visiting all vertices in  $\{v\in[n]: v\equiv_{g_{j-1}^{\phi}} 1\}$.  We ``copy and translate'' this path to have the paths $$v_1+t\cdot \phi(j), v_2+t\cdot \phi(j), ..., v_{n/g_{j-1}^{\phi}}+t\cdot \phi(j), \hspace{5mm} t=0, 1, ..., \frac{g_{j-1}^{\phi}}{g_{j}^{\phi}}-1.$$  We merge these into a single path on  all vertices in $\{v\in[n]: v\equiv_{g_{j-1}^{\phi}} 1\}$ by adding the edges $\{v_{n/g_{j-1}^{\phi}}+t\cdot \phi(j), v_{n/g_{j-1}^{\phi}}+(t+1)\cdot \phi(j)\}$ for $0\leq t \leq \frac{g_{j-1}^{\phi}}{g_{j}^{\phi}}-2, t\equiv_2 0$ and $\{v_{1}+t\cdot \phi(j), v_{1}+(t+1)\cdot \phi(j)\}$ for $1\leq t \leq \frac{g_{j-1}^{\phi}}{g_{j}^{\phi}}-1, t\equiv_2 1$. 

To see that this is a minimum-cost Hamiltonian path, note that for $1\leq i\leq \ell,$  the graph  $C\langle\{\phi(1), \phi(2), ..., \phi(i)\}\rangle$ has $g_i^{\phi}$ components (in particular, the components $\{v\in[n]: v\equiv_{g_{i}^{\phi}} k\}$ for $k=1, ..., g_i^{\phi}.$)  Hence, at most $n-g_i^{\phi}$ edges can be used from the cheapest $i$ stripes without creating a cycle. Kruskal's algorithm will find a minimum-cost spanning tree using $n-g_1^{\phi}=g_0^{\phi}-g_1^{\phi}$ edges from the cheapest stripe, $g_1^{\phi}-g_2^{\phi}$ edges from the second cheapest stripe, and in general $g_{i-1}^{\phi}-g_i^{\phi}$ edges from the $i$th cheapest stripe.  All together, this spanning tree costs  $\sum_{i=1}^{\ell} (g_{i-1}^{\phi}-g_i^{\phi})c_{\phi(i)}.$  Since any Hamiltonian path  is itself a spanning tree,  any Hamiltonian path must cost at least this much; the constructed Hamiltonian path achieves this lower bound and is therefore optimal.   \hfill 
\end{proof}

We can now prove our main theorem about vertices of $EL(n)$, which we restate below:

\mainthm*

To prove this theorem, we will only consider encoding sequences where $s_{k-1}\neq x$. All of these paths have an analogous final step to what's shown in  Figure \ref{fig:BLGEx1}: merging together exactly two paths of length $n/2$, one of which starts at vertex $1$ and ends at some vertex $v$, and its translate starting at vertex $1+\phi(\ell)$ and ending at $v+\phi(\ell)$  (and specifically, $\phi(\ell)=s_{k}$).  The process described in our proof of  Proposition \ref{prop:HP}  adds the edge $\{v, v+\phi(\ell)\}$.  Adding the edge $\{1, 1+\phi(\ell)\}$ uses one additional length-$\phi(\ell)$ edge to extend this to a Hamiltonian cycle.  See, e.g., the Extended Edge-Length Vector column of Table \ref{tab:BLG8}.

\begin{proof}[Proof (of Theorem  \ref{mainthm})]
Proposition \ref{prop:BLGPaths} gives us $2^k \prod_{i=1}^{k-3}(2^i+1)$ BLG paths.  There are at least $2^{k-1} \prod_{i=1}^{k-3}(2^i+1)$ BLG paths whose encoding sequence has $s_{k-1}\neq x$: since $\{1\leq j \leq d: \gcd(n, j)=2\}$ includes 2, we have that $|S_{k-1}|\geq 2$.  Hence $$|S_{k-1}\backslash \{x\}| = |\{1\leq j \leq d: \gcd(n, j)=2\}| = |S_{k-1}|-1 \geq |S_{k-1}|/2.$$  To prove Theorem \ref{mainthm}, we need to show two things: First, that extending these to BLG paths to Hamiltonian cycles gives rise to at least $2^{k-2} \prod_{i=1}^{k-3}(2^i+1)$ distinct  edge-length vectors  of  Hamiltonian cycles.  Second, that each such cycle indeed corresponds to a vertex of $EL(n)$.

First, consider our $2^{k-1} \prod_{i=1}^{k-3}(2^i+1)$ BLG paths with $s_{k-1}\neq x$, which are extended to a Hamiltonian cycle by adding a single length-$s_k$ edge.  When we extend these to edge-length vectors of Hamiltonian Cycles, we claim that we get at least $2^{k-2} \prod_{i=1}^{k-3}(2^i+1)$ distinct edge-length vectors.  To do so, consider any BLG path with edge-length vector $p=(p_1, ..., p_d)$. Recall that Proposition \ref{prop:HP} indicates that $p_{\phi(i)}=g_{i-1}^{\phi}-g_i^{\phi}$ (for a corresponding cost permutation $\phi$).  Similarly, Claim \ref{cm:small} indicates that all elements of the $g$-sequence are (weakly decreasing) powers of 2.  Thus the only non-zero $p_i$ are those that arise are those that occur whenever the $g$-sequence decreases between two powers of 2, and they will all be distinct numbers of the form  the form $2^i-2^{j}$ for $i, j\in \N$ with $i > j \geq 0$.   Similarly, since $s_{k-1}\neq x$, $p_{s_k}=1$; this corresponds to when the $g$-sequence drops from 2 to 1 (specifically,  $g_{\phi^{-1}(s_k)}=1$ because $\gcd(s_k, n)=1$;  $g_{\phi^{-1}(s_k)-1}=2$, because $s_{k-1}\neq x$ means that $\gcd(s_{k-1}, n)=2$ and that edges of length $s_{k-1}$ are cheaper than any edges whose length is in $S_k$).  

When we extend $p$ to an edge-length vector of a Hamiltonian cycle, we add a single edge of length $\phi(\ell)=s_k$.  If $p'=(p'_1, ..., p'_d)$ is the extended edge-length vector of the Hamiltonian cycle, then $$p'_t=\begin{cases} p_t, & t\neq s_k \\ 2, & t = s_k \end{cases}.$$  Moreover, each $p'_t$ is either 0 or of the form $2^i-2^{j}$ for $i, j\in \N$ with $i > j \geq 1$. As before, all non-zero entries of $p'$ will be distinct, except possibly there will two 2's (this happens exactly when $s_{k-2}\neq x$, in which case $p'_{s_{k-1}}$ and $p'_{s_k}$ will be 2).

Thus there is only way for the edge-length vectors of two distinct BLG paths (arising from distinct encoding sequences) to be extended to the same edge length vector of a Hamiltonian cycle: if one path (with edge-length vector $p$) has $p_i=2$ and $p_j=1$, and another path (with edge-length vector $q$) has  $q_i=1$ and $q_j=2$, and $p_t=q_t$ for all other $t\in [d], t\neq i, t\neq j$.  Then, in theory, both could extend to the same cycle using 2 edges of length-$i$ and 2 edges of length-$j$ (and $p_t=q_t$ edges of all other lengths $t$).  It turns out that this can never happen, but it is more succinct to observe that this can only reduce the number of extended edge-length vectors (of Hamiltonian cycles) by a factor of 2.  Hence, we indeed have $2^{k-2} \prod_{i=1}^{k-3}(2^i+1)$ unique edge-length vectors of Hamiltonian cycles.

Finally, we must argue that each resulting extended edge-length vector is indeed a vertex of $EL(n)$. Consider a BLG path with corresponding encoding sequence $s=(s_1, ..., s_k)$.  We claim that the extended edge-length vector is the unique optimal solution to any circulant instance where $\phi(1)$ is the lowest-indexed $s_i$ that's not $x$ (i.e. the first $s_i\neq x$ as you read $s$ from left-to-right), $\phi(2)$ is the second-lowest-indexed $s_i$ that's not $x$, and so on (and all edge-length costs are distinct).

Optimality follows from a greedy-algorithm argument.  The original BLG path uses $(g_{i-1}^{\phi}-g_i^{\phi})$ edges of length $\phi(i)$, for $i=1, ..., \ell$, and the extended Hamiltonian cycle adds exactly one edge of length $\phi(\ell).$  For $1\leq i < \ell$, no Hamiltonian path (or cycle) can use more than $n-g_i^{\phi}$ total edges of lengths $\phi(1), \phi(2), ..., \phi(i)$, as the Circulant graph $C\langle\{\phi(1), ..., \phi(i)\}\rangle$ has $g_i^{\phi}$ components.  Our cycle thus uses as many edges of the cheapest possible edge-length $\phi(1)$ as possible, without forcing a subcycle; then as many combined edges of lengths $\phi(1)$ and $\phi(2)$ as possible (the two cheapest lengths) without creating a subcycle; then as many combined edges of lengths $\phi(1)$, $\phi(2)$, and $\phi(3)$ as possible without creating a subcycle; and so on until using as many combined edges of lengths $\phi(1)$, $\phi(2)$, ... $\phi(\ell)$ as possible without creating a subcycle.  Since all edge costs are distinct, any other Hamiltonian cycle will necessarily cost more.

\hfill
\end{proof}

\section{Conclusions}
Theorem  \ref{mainthm} establishes that the edge-length polytope is intimately linked to the factorization of $n$: there are few enough vertices that a brute-force ``check all vertices'' algorithm is formally efficient whenever $n$ is a prime or a prime-squared, but this is not the case when $n$ is a power of 2.  It would be particularly interesting if the number of vertices of $EL(n)$ was based on the number of factors of $n$.

\section*{Acknowledgments}

\vspace{-3mm}

\noindent The paper was supported by the National Science Foundation, Grant No. CCF-2153331. We thank the referee for careful and helpful feedback.
\vspace{-8mm}

\clearpage

\bibliography{bibliog} 

\begin{thebibliography}{10}

\bibitem{App06b}
D.~L. Applegate, R.~E. Bixby, V.~Chvatál, and W.~J. Cook.
\newblock {\em The Traveling Salesman Problem: A Computational Study}.
\newblock Princeton University Press, 2006.

\bibitem{Aro96}
S.~Arora.
\newblock Polynomial time approximation schemes for {E}uclidean traveling salesman and other geometric problems.
\newblock {\em Journal of the ACM}, 45(5):753--782, 1998.

\bibitem{Avi23}
A.~V. {\'A}vila.
\newblock On unitary/strong linear realizations: Buratti--{H}orak--{R}osa conjecture.
\newblock {\em Bolet{\'\i}n de la Sociedad Matem{\'a}tica Mexicana: Tercera Serie}, 29(3):40, 2023.

\bibitem{Bea24}
A.~Beal, Y.~Bouabida, S.~C. Gutekunst, and A.~Rustad.
\newblock Circulant tsp special cases: Easily-solvable cases and improved approximations.
\newblock {\em Operations Research Letters}, 55:107133, 2024.

\bibitem{Ber08}
P.~Berman and M.~Karpinski.
\newblock 8/7-approximation algorithm for (1, 2)-{T}{S}{P}.
\newblock In {\em Proceedings of the Seventeenth Annual ACM-SIAM Symposium on Discrete Algorithm}, pages 641--648, 2006.

\bibitem{Ber97}
D.~Bertsimas and J.~N. Tsitsiklis.
\newblock {\em Introduction to linear optimization}, volume~6.
\newblock Athena Scientific Belmont, MA, 1997.

\bibitem{Boyd91}
S.~C. Boyd and W.~H. Cunningham.
\newblock Small travelling salesman polytopes.
\newblock {\em Mathematics of Operations Research}, 16(2):259--271, 1991.

\bibitem{Bur13}
M.~Buratti and F.~Merola.
\newblock Dihedral {H}amiltonian cycle systems of the cocktail party graph.
\newblock {\em Journal of Combinatorial Designs}, 21(1):1--23, 2013.

\bibitem{Burk91}
R.~Burkard and W.~Sandholzer.
\newblock Efficiently solvable special cases of bottleneck travelling salesman problems.
\newblock {\em Discrete Applied Mathematics}, 32(1):61--76, 1991.

\bibitem{Burk97}
R.~E. Burkard.
\newblock Efficiently solvable special cases of hard combinatorial optimization problems.
\newblock {\em Mathematical Programming}, 79(1):55--69, 1997.

\bibitem{Burk98}
R.~E. Burkard, V.~G. De{\u\i}neko, R.~{van Dal}, J.~A.~A. {van der Veen}, and G.~J. Woeginger.
\newblock Well-solvable special cases of the traveling salesman problem: A survey.
\newblock {\em SIAM Review}, 40(3):496--546, 1998.

\bibitem{Cap10}
S.~Capparelli and A.~Del~Fra.
\newblock Hamiltonian paths in the complete graph with edge-lengths 1, 2, 3.
\newblock {\em {T}he {E}lectronic {J}ournal of {C}ombinatorics}, 17(1):R44, 2010.

\bibitem{Cha22}
P.~Chand and M.~Ollis.
\newblock The {B}uratti-{H}orak-{R}osa conjecture holds for some underlying sets of size three.
\newblock In S.~Heuss, R.~Low, and J.~Wierman, editors, {\em Southeastern International Conference on Combinatorics, Graph Theory, and Computing}, volume 462, pages 407--427. Springer Proceedings in Mathematics \& Statistics, 2022.

\bibitem{ch91}
T.~Christof, M.~J{\"u}nger, and G.~Reinelt.
\newblock A complete description of the traveling salesman polytope on 8 nodes.
\newblock {\em Operations Research Letters}, 10(9):497--500, 1991.

\bibitem{Chr96}
T.~Christof and G.~Reinelt.
\newblock Combinatorial optimization and small polytopes.
\newblock {\em Top}, 4(1):1--53, 1996.

\bibitem{Chr76}
N.~Christofides.
\newblock Worst-case analysis of a new heuristic for the travelling salesman problem.
\newblock Technical report, Report 388, Graduate School of Industrial Administration, Carnegie-Mellon University, Pittsburgh, PA, 1976.

\bibitem{Chv73}
V.~Chv{\'a}tal.
\newblock Edmonds polytopes and weakly hamiltonian graphs.
\newblock {\em Mathematical programming}, 5(1):29--40, 1973.

\bibitem{Cor85}
G.~Cornu{\'e}jols, J.~Fonlupt, and D.~Naddef.
\newblock The traveling salesman problem on a graph and some related integer polyhedra.
\newblock {\em Mathematical programming}, 33(1):1--27, 1985.

\bibitem{cost18}
S.~Costa, F.~Morini, A.~Pasotti, and M.~A. Pellegrini.
\newblock A problem on partial sums in {A}belian groups.
\newblock {\em Discrete Mathematics}, 341(3):705--712, 2018.

\bibitem{Klerk11}
E.~de~Klerk and C.~Dobre.
\newblock A comparison of lower bounds for the symmetric circulant traveling salesman problem.
\newblock {\em Discrete Applied Mathematics}, 159(16):1815--1826, 2011.

\bibitem{Din09}
J.~H. Dinitz and S.~R. Janiszewski.
\newblock On {H}amiltonian paths with prescribed edge lengths in the complete graph.
\newblock {\em Bulletin of the Institute of Combinatorics and its Applications}, 57:42--52, 2009.

\bibitem{Gar77}
R.~S. Garfinkel.
\newblock Minimizing wallpaper waste, part 1: {A} class of traveling salesman problems.
\newblock {\em Operations Research}, 25(5):741--751, 1977.

\bibitem{Gaw00}
E.~Gawrilow and M.~Joswig.
\newblock Polymake: a framework for analyzing convex polytopes.
\newblock In {\em Polytopes—combinatorics and computation}, pages 43--73. Springer, 2000.

\bibitem{Ger08b}
I.~Gerace and F.~Greco.
\newblock The travelling salesman problem in symmetric circulant matrices with two stripes.
\newblock {\em Mathematical Structures in Computer Science}, 18(1):165--175, 2008.

\bibitem{Gil85}
P.~C. Gilmore, E.~L. Lawler, and D.~B. Shmoys.
\newblock Well-solved special cases.
\newblock In E.~L. Lawler, J.~K. Lenstra, A.~H.~G. {Rinnooy Kan}, and D.~B. Shmoys, editors, {\em The Traveling Salesman Problem: A Guided Tour of Combinatorial Optimization}, pages 87--143. John Wiley and Sons, New York, 1985.

\bibitem{Grec07}
F.~Greco and I.~Gerace.
\newblock The traveling salesman problem in circulant weighted graphs with two stripes.
\newblock {\em Electronic Notes in Theoretical Computer Science}, 169:99--109, 2007.

\bibitem{Gro80}
M.~Gr{\"o}tschel.
\newblock On the monotone symmetric travelling salesman problem: hypohamiltonian/hypotraceable graphs and facets.
\newblock {\em Mathematics of Operations Research}, 5(2):285--292, 1980.

\bibitem{Gro79}
M.~Gr{\"o}tschel and M.~W. Padberg.
\newblock On the symmetric travelling salesman problem i: inequalities.
\newblock {\em Mathematical Programming}, 16(1):265--280, 1979.

\bibitem{Gro86}
M.~Gr{\"o}tschel and M.~W. Padberg.
\newblock Polyhedral theory.
\newblock In E.~L. Lawler, J.~Lenstra, A.~H. G.~R. Kan, and D.~B. Shmoys, editors, {\em The Traveling Salesman Problem: A Guided Tour of Combinatorial Optimization}, pages 251--306. John Wiley \& Sons, New York, 1985.

\bibitem{Gro86b}
M.~Gr{\"o}tschel and W.~R. Pulleyblank.
\newblock Clique tree inequalities and the symmetric travelling salesman problem.
\newblock {\em Mathematics of operations research}, 11(4):537--569, 1986.

\bibitem{gut22}
S.~C. Gutekunst, B.~Jin, and D.~P. Williamson.
\newblock The two-stripe symmetric circulant tsp is in p.
\newblock Draft available at https://people.orie.cornell.edu/dpw/2stripe.pdf, 2021.

\bibitem{Gut19b}
S.~C. Gutekunst and D.~P. Williamson.
\newblock Characterizing the integrality gap of the subtour {LP} for the circulant traveling salesman problem.
\newblock {\em SIAM Journal on Discrete Mathematics}, 33(4):2452--2478, 2019.

\bibitem{Gut20}
S.~C. Gutekunst and D.~P. Williamson.
\newblock The circlet inequalities: A new, circulant-based, facet-defining inequality for the {TSP}.
\newblock {\em Mathematics of Operations Research}, 48(1):393--418, 2023.

\bibitem{Hor09}
P.~Horak and A.~Rosa.
\newblock On a problem of {M}arco {B}uratti.
\newblock {\em The Electronic Journal of Combinatorics}, 16(1):R20, 2009.

\bibitem{kar21}
A.~R. Karlin, N.~Klein, and S.~O. Gharan.
\newblock A (slightly) improved approximation algorithm for metric tsp.
\newblock In {\em Proceedings of the 53rd Annual ACM SIGACT Symposium on Theory of Computing}, pages 32--45, 2021.

\bibitem{Karp12}
M.~Karpinski and R.~Schmied.
\newblock On approximation lower bounds for {TSP} with bounded metrics.
\newblock {\em Electronic Colloquium on Computational Complexity {(ECCC)}}, 19:8, 2012.

\bibitem{Law07}
E.~L. Lawler, J.~K. Lenstra, A.~H.~G. {Rinnooy Kan}, and D.~B. Shmoys.
\newblock {\em The Traveling Salesman Problem: A Guided Tour of Combinatorial Optimization}.
\newblock John Wiley and Sons, 1985.

\bibitem{Mc22}
B.~D. McKay and T.~Peters.
\newblock Paths through equally spaced points on a circle.
\newblock {\em Journal of Integer Sequences}, 25(7)-22.7.4, 2022.

\bibitem{Med93}
E.~Medova.
\newblock Using {QAP} bounds for the circulant {TSP} to design reconfigurable networks.
\newblock In {\em Quadratic Assignment and Related Problems, Proceedings of a {DIMACS} Workshop, New Brunswick, New Jersey, USA, May 20-21, 1993}, pages 275--292, 1993.

\bibitem{Mitch99}
J.~S.~B. Mitchell.
\newblock Guillotine subdivisions approximate polygonal subdivisions: A simple polynomial-time approximation scheme for geometric {TSP}, k-{MST}, and related problems.
\newblock {\em SIAM Journal on Computing}, 28(4):1298--1309, 1999.

\bibitem{Mom16}
T.~M{\"{o}}mke and O.~Svensson.
\newblock Removing and adding edges for the traveling salesman problem.
\newblock {\em Journal of the ACM}, 63(1):2:1--2:28, 2016.

\bibitem{Muc14}
M.~Mucha.
\newblock $\f{13}{9}$-approximation for graphic {TSP}.
\newblock {\em Theory of Computing Systems}, 55(4):640--657, 2014.

\bibitem{Nad92b}
D.~Naddef.
\newblock The binested inequalities for the symmetric traveling salesman polytope.
\newblock {\em Mathematics of Operations Research}, 17(4):882--900, 1992.

\bibitem{Nad06}
D.~Naddef.
\newblock Polyhedral theory and branch-and-cut algorithms for the symmetric {TSP}.
\newblock In G.~Gutin and A.~P. Punnen, editors, {\em The traveling salesman problem and its variations}, pages 29--116. Springer, 2007.

\bibitem{Nad91}
D.~Naddef and G.~Rinaldi.
\newblock The symmetric traveling salesman polytope and its graphical relaxation: Composition of valid inequalities.
\newblock {\em Mathematical Programming}, 51(1-3):359--400, 1991.

\bibitem{Nad92}
D.~Naddef and G.~Rinaldi.
\newblock The crown inequalities for the symmetric traveling salesman polytope.
\newblock {\em Mathematics of Operations Research}, 17(2):308--326, 1992.

\bibitem{Oll21}
M.~Ollis, A.~Pasotti, M.~A. Pellegrini, and J.~R. Schmitt.
\newblock New methods to attack the {B}uratti-{H}orak-{R}osa conjecture.
\newblock {\em Discrete Mathematics}, 344(9):112486, 2021.

\bibitem{Oll22}
M.~Ollis, A.~Pasotti, M.~A. Pellegrini, and J.~R. Schmitt.
\newblock Growable realizations: a powerful approach to the {B}uratti-{H}orak-{R}osa conjecture.
\newblock {\em Ars Mathematica Contemporanea}, 22(4)-4, 2022.

\bibitem{Gha11}
S.~{Oveis Gharan}, A.~Saberi, and M.~Singh.
\newblock A randomized rounding approach to the traveling salesman problem.
\newblock In {\em Proceedings of the 52nd Annual {IEEE} Symposium on the Foundations of Computer Science}, pages 550--559, 2011.

\bibitem{Pap93}
C.~H. Papadimitriou and M.~Yannakakis.
\newblock The traveling salesman problem with distances one and two.
\newblock {\em Mathematics of Operations Research}, 18(1):1--11, 1993.

\bibitem{Pas14}
A.~Pasotti and M.~A. Pellegrini.
\newblock A new result on the problem of {B}uratti, {H}orak and {R}osa.
\newblock {\em Discrete Mathematics}, 319:1 -- 14, 2014.

\bibitem{Pas14b}
A.~Pasotti and M.~A. Pellegrini.
\newblock On the {B}uratti-{H}orak-{R}osa conjecture about {H}amiltonian paths in complete graphs.
\newblock {\em The Electronic Journal of Combinatorics}, 21(2)-2.30, 2014.

\bibitem{Pas19}
A.~Pasotti and M.~A. Pellegrini.
\newblock Further progress on the {B}uratti-{H}orak-{R}osa conjecture.
\newblock {\em arXiv preprint arXiv:1912.07377}, 2019.

\bibitem{Seb14}
A.~Seb{\H{o}} and J.~Vygen.
\newblock Shorter tours by nicer ears: 7/5-approximation for the graph-{TSP}, 3/2 for the path version, and 4/3 for two-edge-connected subgraphs.
\newblock {\em Combinatorica}, 34(5):597--629, 2014.

\bibitem{Serd78}
A.~Serdyukov.
\newblock On some extremal walks in graphs.
\newblock {\em Upravlyaemye Sistemy}, 17:76--79, 1978.

\bibitem{Slo}
N.~J.~A. Sloane.
\newblock The {O}n-{L}ine {E}ncyclopedia of {I}nteger {S}equences.
\newblock http://oeis.org.

\bibitem{Vaz22}
A.~V{\'a}zquez-Avila.
\newblock A note on the {B}uratti-{H}orak-{R}osa conjecture about hamiltonian paths in complete graphs.
\newblock {\em Bull. ICA}, 94:53--70, 2022.

\bibitem{DDBook}
D.~P. Williamson and D.~B. Shmoys.
\newblock {\em The {D}esign of {A}pproximation {A}lgorithms}.
\newblock Cambridge University Press, New York, 2011.

\bibitem{Wol99}
L.~A. Wolsey and G.~L. Nemhauser.
\newblock {\em Integer and Combinatorial Optimization}, volume~55.
\newblock John Wiley \& Sons, 1999.

\end{thebibliography}
\bibliographystyle{abbrv}

\begin{appendix}

\section{Vertices of $EL(n)$ for Small $n$}\label{sec:app}
In this Appendix, we include computational data from Polymake giving the vertices  of $EL(n)$ for small $n$.

\subsection{$n=7$}
When $n=7$, there are 3 vertices: (7, 0, 0), (0, 7, 0), (0, 0, 7).

\subsection{$n=8$}
When $n=8$, there are 10 vertices: (8, 0, 0, 0), (0, 0, 8, 0), (6, 0, 0, 2), (2, 6, 0, 0), (0, 6, 2, 0), (0, 0, 6, 2), (1, 0, 3, 4),
(3, 0, 1, 4), (2, 2, 0, 4), (0, 2, 2, 4).

\subsection{$n=9$}
When $n=9$, there are 6 vertices: (9, 0, 0, 0), (0, 9, 0, 0), (0, 0, 0, 9), (3, 0, 6, 0), (0, 3, 6, 0), (0, 0, 6, 3)

\subsection{$n=10$}
When $n=10$, there are 18 vertices: (10, 0, 0, 0, 0), 
(0, 0, 10, 0, 0), 
(5, 0, 0, 0, 5),
(0, 0, 5, 0, 5),
(2, 8, 0, 0, 0),
(0, 8, 2, 0, 0),
(0, 8, 0, 0, 2),
(2, 0, 0, 8, 0),
(0, 0, 0, 8, 2),
(0, 0, 2, 8, 0),
(1, 0, 0, 4, 5),
(0, 4, 1, 0, 5),
(0, 5, 0, 1, 4),
(0, 1, 0, 5, 4),
(3, 2, 0, 0, 5),
(0, 0, 3, 2, 5),
(1, 3, 0, 1, 5),
(0, 1, 1, 3, 5).

\subsection{$n=11$}
When $n=11$, there are 5 vertices: (11, 0, 0, 0, 0), 
(0, 11, 0, 0, 0),
(0, 0, 11, 0, 0),
(0, 0, 0, 11, 0),
(0, 0, 0, 0, 11).

\subsection{$n=12$}
When $n=12$, there are 48 vertices: (12, 0, 0, 0, 0, 0),
(0, 0, 0, 0, 12, 0),
(10, 0 ,0, 0, 0, 2),
(0, 0, 0, 0, 10, 2),
(2, 10, 0, 0, 0, 0),
(0, 10, 2, 0, 0, 0),
(0, 10, 0, 0, 2, 0),
(3, 0, 9, 0, 0, 0),
(0, 0, 9, 0, 3, 0),
(4, 0, 0, 8, 0, 0),
(0, 0, 4, 8, 0, 0),
(0, 0, 0, 8, 4, 0),
(0, 4, 8, 0, 0, 0),
(0, 0, 8, 4, 0, 0),
(3, 0, 3, 0, 0, 6),
(0, 0, 3, 0, 3, 6),
(0, 3, 8, 0, 0, 1),
(0, 0, 8, 3, 0, 1),
(0, 3, 4, 0, 0, 5),
(0, 0, 4, 3, 0, 5),
(2, 0, 0, 4, 0, 6),
(0, 0, 2, 4, 0, 6),
(0, 0, 0, 4, 2, 6),
(2, 4, 0, 0, 0, 6),
(0, 4, 0, 0, 2, 6),
(0, 4, 2, 0, 0, 6),
(4, 0, 2, 0, 0, 6),
(0, 0, 2, 0, 4, 6),
(2, 0, 0, 8, 0, 2),
(0, 0, 2, 8, 0, 2),
(0, 0, 0, 8, 2, 2),
(2, 2, 0, 8, 0, 0),
(0, 2, 2, 8, 0, 0),
(0, 2, 0, 8, 2, 0),
(4, 1, 0, 1, 0, 6),
(0, 1, 0, 1, 4, 6),
(6, 1, 0, 0, 0, 5),
(1, 0, 0, 0, 5, 6),
(5, 0, 0, 0, 1, 6),
(0, 1, 0, 0, 6, 5),
(1, 2, 9, 0, 0, 0),
(0, 2, 9, 0, 1, 0),
(1, 0, 9, 2, 0, 0),
(0, 0, 9, 2, 1, 0),
(1, 0, 3, 2, 0, 6),
(1, 2, 3, 0, 0, 6),
(0, 2, 3, 0, 1, 6),
(0, 0, 3, 2, 1, 6).

\subsection{$n=13$}
When $n=13$, there are 6 vertices: (13, 0, 0, 0, 0, 0),
(0, 13, 0, 0, 0, 0),
(0, 0, 13, 0, 0, 0),
(0, 0, 0, 13, 0, 0),
(0, 0, 0, 0, 13, 0),
(0, 0, 0, 0, 0, 13).

\subsection{$n=14$}
When $n=14$, there are 51 vertices:
(14,0,0,0,0,0,0),
(2,12,0,0,0,0,0),
(1,8,0,0,0,0,5),
(2,0,0,12,0,0,0),
(2,0,0,0,0,12,0),
(1,0,0,0,0,6,7),
(2,0,0,4,1,0,7),
(1,0,0,5,0,1,7),
(5,0,0,2,0,0,7),
(1,2,0,4,0,0,7),
(1,4,2,0,0,0,7),
(1,5,0,0,0,1,7),
(1,4,0,2,0,0,7),
(3,4,0,0,0,0,7),
(7,0,0,0,0,0,7),
(0,12,2,0,0,0,0),
(0,12,0,0,2,0,0),
(0,12,0,0,0,0,2),
(0,2,0,6,0,0,6),
(0,6,0,0,0,2,6),
(0,7,0,1,0,0,6),
(0,1,0,0,0,7,6),
(0,1,0,0,1,5,7),
(0,2,1,0,0,4,7),
(0,4,1,0,0,2,7),
(0,6,0,0,1,0,7),
(0,1,0,5,1,0,7),
(0,5,1,1,0,0,7),
(0,2,5,0,0,0,7),
(0,0,14,0,0,0,0),
(0,0,2,12,0,0,0),
(0,0,2,0,0,12,0),
(0,0,1,0,0,8,5),
(0,0,1,0,2,4,7),
(0,0,1,1,0,5,7),
(0,0,3,0,0,4,7),
(0,0,1,6,0,0,7),
(0,0,7,0,0,0,7),
(0,0,0,8,1,0,5),
(0,0,0,7,0,1,6),
(0,0,0,2,0,6,6),
(0,0,0,12,0,0,2),
(0,0,0,12,2,0,0),
(0,0,0,2,1,4,7),
(0,0,0,4,1,2,7),
(0,0,0,4,3,0,7),
(0,0,0,0,14,0,0),
(0,0,0,0,2,12,0),
(0,0,0,0,5,2,7),
(0,0,0,0,7,0,7),
(0,0,0,0,0,12,2).

\end{appendix}

\end{document}